\documentclass[12pt]{article}
\usepackage{cite}
\usepackage{amsmath}
\usepackage{latexsym}
\usepackage{amsfonts}
\usepackage[normalem]{ulem}
\usepackage{soul}
\usepackage{array}
\usepackage{amssymb}
\usepackage{extarrows}
\usepackage{graphicx}

\usepackage{subfig}
\usepackage{wrapfig}
\usepackage{wasysym}
\usepackage{enumitem}
\usepackage{adjustbox}
\usepackage{ragged2e}
\usepackage[svgnames,table]{xcolor}
\usepackage{tikz}
\usepackage{longtable}
\usepackage{changepage}
\usepackage{setspace}
\usepackage{hhline}
\usepackage{multicol}
\usepackage{tabto}
\usepackage{float}
\usepackage{multirow}
\usepackage{makecell}
\usepackage{fancyhdr}
\usepackage[toc,page]{appendix}
\usepackage[hidelinks]{hyperref}
\usetikzlibrary{shapes.symbols,shapes.geometric,shadows,arrows.meta}
\tikzset{>={Latex[width=1.5mm,length=2mm]}}
\usepackage{flowchart}\usepackage[paperheight=11.0in,paperwidth=8.5in,left=1.0in,right=1.0in,top=1.0in,bottom=1.0in,headheight=1in]{geometry}
\usepackage[utf8]{inputenc}
\usepackage[T1]{fontenc}
\usepackage{xfrac}
\usepackage{scrextend}

\TabPositions{0.5in,1.0in,1.5in,2.0in,2.5in,3.0in,3.5in,4.0in,4.5in,5.0in,5.5in,6.0in,}

\renewcommand{\_}{\kern-1.5pt\textunderscore\kern-1.5pt}

\setlistdepth{9}
\renewlist{enumerate}{enumerate}{9}
		\setlist[enumerate,1]{label=\arabic*)}
		\setlist[enumerate,2]{label=\alph*)}
		\setlist[enumerate,3]{label=(\roman*)}
		\setlist[enumerate,4]{label=(\arabic*)}
		\setlist[enumerate,5]{label=(\Alph*)}
		\setlist[enumerate,6]{label=(\Roman*)}
		\setlist[enumerate,7]{label=\arabic*}
		\setlist[enumerate,8]{label=\alph*}
		\setlist[enumerate,9]{label=\roman*}

\renewlist{itemize}{itemize}{9}
		\setlist[itemize]{label=$\cdot$}
		\setlist[itemize,1]{label=\textbullet}
		\setlist[itemize,2]{label=$\circ$}
		\setlist[itemize,3]{label=$\ast$}
		\setlist[itemize,4]{label=$\dagger$}
		\setlist[itemize,5]{label=$\triangleright$}
		\setlist[itemize,6]{label=$\bigstar$}
		\setlist[itemize,7]{label=$\blacklozenge$}
		\setlist[itemize,8]{label=$\prime$}

\usepackage{xfrac}
\usepackage{amsfonts}
\usepackage{amsmath}
\usepackage{amssymb}
\usepackage{amsthm}
\usepackage{graphicx}
\usepackage{units}
\usepackage{algorithm}
\usepackage{algorithmicx}
\usepackage[noend]{algpseudocode}
\usepackage{setspace}
\usepackage{float}
\usepackage{kantlipsum}
\usepackage{multicol}
\usepackage{xspace}
\usepackage{wrapfig}

\usepackage{etoolbox}
\makeatletter
\patchcmd{\@maketitle}
  {\newpage\null\vskip 2em}{}
  {}{}
\makeatother

\author{Siddhartha Jayanti}
\date{January, 2016}

\makeatletter
\def\BState{\State\hskip-\ALG@thistlm}
\makeatother

\algnewcommand\Not{\textbf{not }}
\algdef{SE}[DOWHILE]{Do}{doWhile}{\algorithmicdo\xspace}[1]{\algorithmicuntil\ #1}
\algdef{SN}{DoTwice}{EndDoTwice}[1]{\textbf{do twice} \(\mbox{#1}\)}{}%

\newcommand{\floor}[1]{\left\lfloor #1 \right\rfloor}

\newcommand{\Find}[1]{\textit{find}(#1)} 
\newcommand{\Unite}[1]{\textit{unite}(#1)}

\newcommand{\CAS}[1]{CAS(#1)}

\renewcommand{\ast}[0]{{}^*} 

\begin{document}

\newtheorem{theorem}{Theorem}
\newtheorem{example}{Example}
\newtheorem{lemma}{Lemma}
\newtheorem{corollary}{Corollary}
\newtheorem{remark}{Remark}
\newtheorem{conjecture}{Conjecture}
\newtheorem*{thm}{Theorem}

\begin{Center}
\textbf{{\large Concurrent Disjoint Set Union}\footnote{This paper is a combination and revision of \cite{JT16} and part of \cite{JayantiTarjanBoix}.  Research at Princeton University partially supported by the Princeton University Innovation Fund for Industrial Collaborations and gifts from Microsoft.}}
\end{Center}\par

\vspace{\baselineskip}
\begin{Center}
Siddhartha Jayanti\footnote{Computer Science and Artificial Intelligence Laboratory (CSAIL), MIT} and Robert E. Tarjan\footnote{Department of Computer Science, Princeton University; and Intertrust Technologies}
\end{Center}\par

\begin{Center}
March 2020 
\end{Center}\par

\vspace{1in}
\begin{abstract}

We develop and analyze concurrent algorithms for the disjoint set union (``union-find" ) problem in the shared memory, asynchronous multiprocessor model of computation,\ with CAS (compare and swap) or DCAS (double compare and swap) as the synchronization primitive.  \  We give a deterministic bounded wait-free algorithm that uses DCAS and has a total work bound of
$O\biggl( m \cdot \left( \log{\left(\frac{np}{m} + 1 \right)} + \alpha{\left(n, \frac{m}{np} \right)} \right) \biggr)$ 
for a problem with $n$ elements and $m$ operations solved by $p$ processes, where $\alpha$ is a functional inverse of Ackermann’s function.\  We give two randomized algorithms that use only CAS and have\ the same work bound in expectation.  The analysis of the second randomized algorithm is valid even if the scheduler is adversarial.\  Our DCAS and randomized algorithms take O(log\textit{n}) steps per operation, worst-case for the DCAS algorithm, high-probability for the randomized algorithms.\  Our work and step bounds grow only logarithmically with $p$, making our algorithms truly scalable.\  We prove that for a class of symmetric algorithms that includes ours, no better step or work bound is possible.\  \  \par
\end{abstract}

\newpage
\section{Introduction}

\vspace{\baselineskip}
As data sets get bigger and bigger, it becomes more and more important to harness the potential of parallelism to solve computational problems -\ even linear time is too slow.  In the late twentieth century, many beautiful and efficient algorithms were developed in the PRAM (parallel random access machine)\ model, which assumes a memory shared among many synchronized processors.  In practice, however, synchronization is expensive or may not be possible.\  A weaker model that has attracted much attention in the distributed systems community is the APRAM (asynchronous parallel random access\ machine) model, in which a common memory is shared among many unsynchronized processors.  In the most general version of this model, any processor can be arbitrarily slow compared to any other.\  \par

\vspace{\baselineskip}
Obtaining efficiency bounds in the APRAM model is extremely challenging: the use of locks, for example, seems to make it impossible to guarantee efficiency, since one process could set a lock and then go to sleep indefinitely, blocking progress by any other process that needs access to the same resource. To overcome this problem, systems researchers have invented synchronization primitives that do not use locks, notably CAS (compare and swap) \cite{H}, transactional memory \cite{Herlihy:1993:TMA:173682.165164}, and others.\  These primitives allow at least\ the possibility of obtaining good efficiency bounds for asynchronous concurrent algorithms.  Yet, except for $``$embarrassingly parallel$"$  computations, this possibility is almost unrealized.\  Indeed, we know of only one example of a concurrent data structure (other than our work, to be described) for which a work bound without a term at least linear in the number of processes has been obtained.\  This is an implementation by Ellen and Woefel \cite{EW} of a fetch-and-increment object.\par

\vspace{\baselineskip}
An important problem in data structures that could benefit from an efficient concurrent algorithm is \textit{disjoint set union}, also known as the \textit{union-find} problem.\  The simplest version of this problem requires maintaining a collection of disjoint sets, each containing a unique element called its \textit{representative}, under two operations:\par

\vspace{\baselineskip}
\textit{find}(\textit{x}): return the representative of the set containing element \textit{x.}\par

\vspace{\baselineskip}
\textit{unite}(\textit{x}, \textit{y}): if elements\textit{ x }and \textit{y} are in different sets, unite these sets into a single set and designate some element in the new set to be its representative; otherwise, do nothing.\par

\vspace{\baselineskip}
Each initial set is a singleton, whose representative is its only element.\  Note that the implementation is free to choose the representative of each new set produced by a \textit{unite}. This freedom simplifies concurrent implementation, as we discuss in Section 4. Other versions of the problem add operations for initializing singleton sets and for maintaining and retrieving information about the sets such as names or sizes.\  We study the simplest version but comment on extensions in Section 9.\par

\vspace{\baselineskip}
Applications of sequential disjoint set union include storage allocation in compilers \cite{LA}, finding minimum spanning trees using Kruskal’s algorithm \cite{K56}, maintaining the connected components of an undirected graph under edge additions \cite{T75}, testing percolation \cite{SW}, finding loops and dominators in flow graphs \cite{tarjan1973testing, tarjan1974finding, FGMT13}, and finding strong components in directed graphs.  
Some of these applications, notably finding connected components \cite{TarjanDFS, SHILOACH198257, Johnson1997connected, HALPERIN20011, Rastogi12, LiuTarjan} and finding strong components, are on immense graphs and could potentially benefit from the use of concurrency to speed up the computation.
For example, model checking requires finding strong components in huge, implicitly defined directed graphs \cite{W02, Blo15, BLvP16}.\  There are sequential linear-time strong components algorithms \cite{TarjanDFS, SHARIR198167}, but these may not be fast enough for this application.
The sequential algorithms use depth-first search \cite{TarjanDFS}, which apparently cannot be efficiently parallelized \cite{Reif}.
If one had an efficient concurrent disjoint set union algorithm one could use it in combination with breadth-first search to potentially speed up model checking. 
This application, described to the second author by Leslie Lamport, was the original motivation for our work.  \par

\vspace{\baselineskip}
The classical sequential solution to the disjoint set union problem is the \textit{compressed tree} data structure \cite{galler1964improved, fischer1972efficiency, hopcroft1973set, Tar75, TvL84}.\  With appropriate tree linking and path compaction rules, \textit{m} operations on sets containing a total of \textit{n} elements take $O(\alpha(n, m/n))$ time \cite{Tar75, TvL84, GKLT14}, where $\alpha$  is a functional inverse of Ackermann’s function, defined in Section 3.\  Three linking rules that suffice are linking by size \cite{Tar75}, linking by rank \cite{TvL84}, and linking by random\ index \cite{GKLT14}; three compaction rules that suffice are compression \cite{Tar75, TvL84, GKLT14}, splitting \cite{TvL84, GKLT14}, and halving \cite{TvL84, GKLT14}.  \par

\vspace{\baselineskip}
Perhaps surprisingly, there has been almost no previous research on wait-free concurrent disjoint set\ union.  We have found only one such effort, that of Anderson and Woll \cite{AW91}.  Their work contains a number of significant ideas that are the genesis of our results, but it has many flaws that reveal the subtlety\ of the problem.  We use their concurrency model.\  In one of our linking algorithms we use DCAS (double compare and swap), as a synchronization primitive, whereas they used only the weaker CAS (compare and swap) primitive.\par

\vspace{\baselineskip}
Anderson and Woll considered an alternative formulation of the problem in which sets do not have representatives and the two operations are \textit{same-set}(\textit{x}, \textit{y}), which returns true if \textit{x} and \textit{y} are in the same set and false otherwise, and \textit{unite}(\textit{x}, \textit{y}), which combines the sets containing \textit{x} and \textit{y}\ into a single set if these sets are different.  (We discuss \textit{same-set} further in Section 4.)\  They attempted to develop an efficient concurrent solution that combines linking by rank with a concurrent version of path halving.\  They claimed a bound of $O(m \cdot (p + \alpha(m,1))$ on the total work, where $p$\ is the number of processors.  (They did not treat \textit{n} as a separate parameter.). Their linking method can produce paths of $\Omega(p)$ nodes of equal rank.\  The $O(mp)$ term in their work bound accounts for such paths.\  Their proof of their upper bound is not correct, because they did not consider\ interference among different processes doing halving on intersecting paths.  Whether or not their bound is correct, it is easy to show that their algorithm can take $\Omega(np)$ work to do $n - 1$ unite operations, compared to the $O(n \alpha(n, 1))$ time required by one process.  Thus in the worst case their work bound gives essentially no speedup.\par

\vspace{\baselineskip}
Anderson and Woll also claimed a work bound of $O(m\cdot(\alpha(m, 1) + \log^* p))$ for a synchronous PRAM algorithm that uses deterministic coin tossing\ \cite{CV86} to break up long paths of equal-rank nodes.  They provided no details of this algorithm and no proof of the work bound.  We think that their bound is incorrect and that the work bound of their algorithm is $\Omega(n \log p)$, since it is easy to construct sets of operations that do linking by rank exactly but such that concurrent finds with halving take $\Omega(\log p)$ steps per find, even on a PRAM.  See Section 8.\  Deterministic coin tossing does seem to be a good idea, however: we conjecture that it can give an efficient (but complicated) deterministic set union algorithm in the APRAM model using only CAS for synchronization, at the cost of a multiplicative $\log^* p$ factor in the work bound.

\vspace{\baselineskip}
In this paper we apply the ideas of Anderson and Woll and some additional ones to develop\ several efficient concurrent algorithms for disjoint set union.  We give three concurrent implementations of \textit{unite},\ one deterministic and the other two randomized.  The deterministic method uses DCAS to do linking by rank.\  The randomized methods use only CAS: one does linking\ by random index, the other does randomized linking by rank.  We also give two concurrent implementations of path splitting, \textit{one-try} and \textit{two-try} \textit{splitting}.\  The former is simpler, but we are able to prove slightly better bounds for the latter, bounds that we think are tight for the problem.\par

\vspace{\baselineskip}
We prove that any of our linking methods in combination with one-try splitting does set union in $O\biggl( m \cdot \left( \log{\left(\frac{np^2}{m} + 1 \right)} + \alpha{\left(n, \frac{m}{np^2} \right)} \right) \biggr)$ work, and in combination with two-try splitting in $O\biggl( m \cdot \left( \log{\left(\frac{np}{m} + 1 \right)} + \alpha{\left(n, \frac{m}{np} \right)} \right) \biggr)$ work.  Each set operation takes $O(\log n)$ steps.  These bounds are worst-case for deterministic linking and high-probability for randomized linking.  The $O(\log n)$ step bound per operation holds even without path splitting; without splitting, the work bound is $O(m \log n)$.  The work and step bounds for randomized linking by rank hold even for an adversarial scheduler, provided that scheduling is based only on information sent to the scheduler, or we allow a form of CAS that writes a random bit.\  The work and step bounds for linking by random index hold provided that the randomization is independent of the order in which\ the unite operations are executed, or, more precisely, independent of the $``$linearization order$"$  of the unite operations.  (We define linearization order in Section 2.)\  We also show that $\Omega\biggl( m \cdot \left( \log{\left(\frac{np}{m} + 1 \right)} + \alpha{\left(n, \frac{m}{np} \right)} \right) \biggr)$ work is needed in the worst case for any algorithm satisfying a symmetry assumption, which implies that our work bound for two-try splitting is best possible for such algorithms.\ \  \par

\vspace{\baselineskip}
The remainder of our paper contains 8 sections.\  Section\ 2\ describes\ our concurrency model.  Section 3 describes the compressed tree data structure and sequential algorithms for disjoint set union.  Section 4 presents concurrent linking by index, a special case of which is concurrent linking by random index, and one-try and two-try splitting.  Section 5 presents preliminary versions of deterministic and randomized linking by rank.\ \ These\ versions\ rely on some simplifying assumptions that we eliminate in Section 6.  Section 7 gives upper bounds on the total work of our algorithms.  Section 8 presents lower bounds.  Section 9 contains some final remarks and open problems.

\vspace{\baselineskip}
\section{Concurrency Model}

\vspace{\baselineskip}
Our concurrency model is the same as that of Anderson and Woll: a shared memory\ multiprocessor, otherwise known as an asynchronous random-access machine (APRAM).  We assume that $p$ processes run concurrently but asynchronously, each doing a different set operation.\ \ Each process has a private memory.  In addition, all processes have access to a shared memory that supports concurrent reads but not concurrent writes.\par

\vspace{\baselineskip}
To provide synchronization of writes to shared memory, we use the \textit{compare and swap} primitive CAS(\textit{x}, \textit{y}, \textit{z}).\  Given the address \textit{x} of a block of shared memory and two values \textit{y} and\textit{ z}, this operation tests whether block \textit{x }holds value \textit{y}; if so, it stores value \textit{z }in block \textit{x} (overwriting \textit{y}) and returns true; if not, it\ returns false.  We also consider the two-block extension DCAS(\textit{u}, \textit{v}, \textit{w}, \textit{x}, \textit{y}, \textit{z}).\  Given the addresses \textit{u} and \textit{x }of two blocks of shared memory and four values \textit{v}, \textit{w}, \textit{y}, and\textit{ z}, this operation tests whether block \textit{u} holds value \textit{v} and block \textit{x }holds value \textit{y}; if both are true, it stores value \textit{w} in block \textit{u} and value \textit{z} in block \textit{x} and returns true;\ if not, it returns false.  These operations are atomic: once one starts, it completes before any other operation can read, write, CAS, or DCAS the affected block or blocks.\  Although both CAS and DCAS return a value indicating success or failure, many of our algorithms do not actually use these values.\par

\vspace{\baselineskip}
In one version of our randomized linking algorithm we use the following randomized version of CAS: atomic operation CAS(\textit{x}, null, flip) tests whether the value of \textit{x }is null and, if so, sets the value of \textit{x}\ equal to true or false, each with probability $\sfrac{1}{2}$.  Such a randomized atomic write operation has been used in algorithms for achieving consensus \cite{Chor87}.\par

\vspace{\baselineskip}
Many current hardware designs include CAS as an instruction; DCAS was supported on the Motorola 68030 \cite{DCAS} but not on any current hardware, as far as we know.\  As we demonstrate in Section 5, it is straightforward to implement linking by rank using DCAS, but much harder using only CAS.\par

\vspace{\baselineskip}
We study concurrent algorithms for disjoint set union that are \textit{linearizable} \cite{HW90} and \textit{bounded} \textit{wait-free} \cite{H91}. To be linearizable means that (i) the outcome of a concurrent execution is the same as if each set operation were executed instantaneously at some distinct time (its linearization time) during its actual execution and (ii) the sequential execution sequence given by the linearization times is correct; that is, all find operations produce answers that are correct at their linearization times.\  The linearization times define a total order of the operations, called the \textit{linearization order}.\ \ Although we focus on linearizable algorithms, some applications of disjoint set union may not require linearizability for correctness.  
We briefly discuss this issue in Section 9, and leave further investigation as an open problem.\par

\vspace{\baselineskip}
To be bounded\ wait-free means that every operation finishes in a bounded number of its own steps.  The \textit{total work} done by a concurrent solution is the total number of steps done by all processes to complete all operations.\  \par

\vspace{\baselineskip}
Two weaker progress properties than bounded wait-freedom are\textit{ wait-freedom} and \textit{lock-freedom} \cite{HerlihyShavitBook}.
A concurrent solution is wait-free if every process\ is guaranteed to finish.  It is\ lock-free if every operation can execute its next step when it chooses to do so, and at least one process is guaranteed to finish its operation.  In general\ a\ lock-free\ solution need not be wait-free, and a wait-free solution need not be bounded wait-free.  In our version of disjoint set union, the number of elements is fixed, which makes it easy to guarantee bounded wait-freedom.  This remains true if we add an operation that allows the creation of singleton sets containing new elements, as long as the total number of set operations is bounded.  If we allow an unbounded number of singleton sets to be created, then our solutions are no longer wait-free, but they remain lock-free.\  In this case there are no meaningful work bounds.\  \par

\section{Data Structure and Sequential Algorithms}

\vspace{\baselineskip}
Our concurrent disjoint set union algorithms use the same data structure as the best sequential algorithms: a \textit{compressed forest}.\  This forest contains one rooted tree per set, whose nodes are the elements of the set and whose root is the set representative.\  Each node \textit{x} has a pointer \textit{x.p}, to its parent if it has a parent or to itself if it is a root.\  The root of the tree is the representative of the set.\par

\vspace{\baselineskip}
The sequential algorithm for \textit{find}(\textit{x}) follows parent pointers from \textit{x }until reaching a node \textit{v} that points to itself, optionally \textit{compacts} the find path (the path of ancestors from \textit{x} to \textit{v}) by replacing the parent of one or more nodes on the find path by a proper ancestor of its parent, and returning \textit{v}.\  \textit{Naïve find} does no compaction.\  Three good compaction rules are \textit{compression}, \textit{splitting}, and \textit{halving}.\  Compression replaces the parent of every node on the find path by the root \textit{v}.\  Splitting replaces the parent of every node on the\ find path by its grandparent.  Halving replaces the parent of every other node on the find path by its grandparent, starting with \textit{x}.\  See Figure 1.\par

\begin{figure}[h]
    \centering
    \includegraphics[scale=0.5]{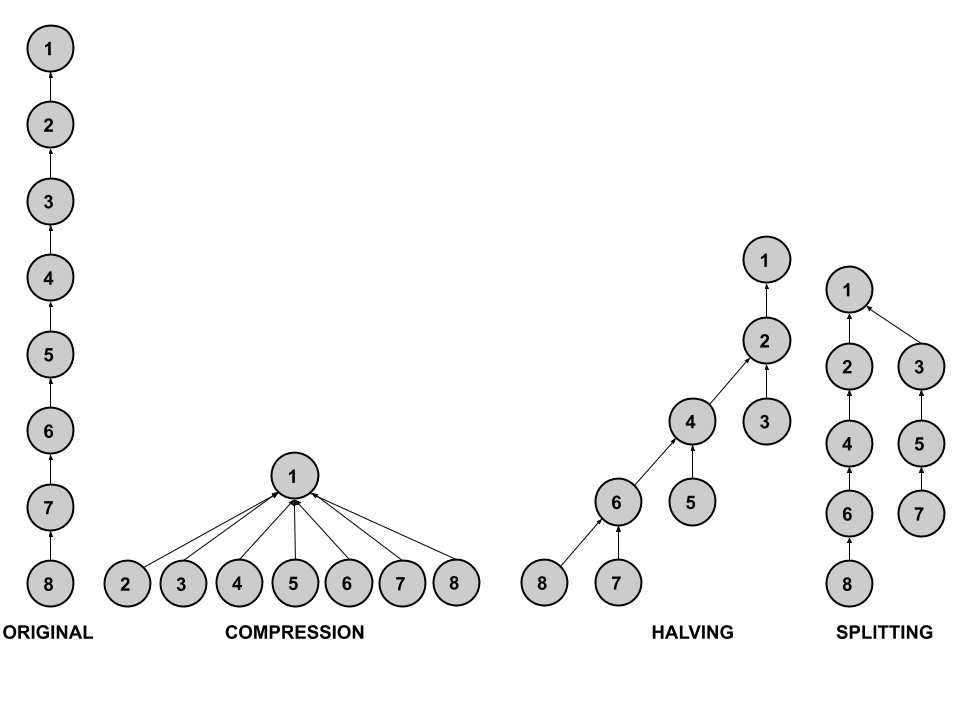}
    \caption{The results of running \textit{find}(8) on the original path with the three different types of compaction:
    compression links all the nodes on the find path directly to the root, thereby ``compressing'' the path;
    halving links alternating nodes on the find path to their grandparents, thereby creating a path of ``half'' the length with nodes hanging off;
    splitting links every node on the find path to its grandparent, thereby ``splitting'' one path into two.}
    \label{fig:1}
\end{figure}

\vspace{\baselineskip}
The sequential implementation of \textit{unite}(\textit{x}, \textit{y}) does \textit{find}(\textit{x}) and \textit{find}(\textit{y}), returning the roots \textit{u} and \textit{v} of the trees containing \textit{x} and \textit{y}, respectively, and tests whether \textit{u} = \textit{v}.\  If $u \ne v$, it \textit{links} \textit{u }and \textit{v }by making one the parent of the other.\  Three good linking rules are \textit{linking by size}, \textit{linking by rank}, and \textit{linking by random index}.\  \textit{Linking by size} maintains the size (number of nodes) of each tree in its root, and makes the root of the tree of larger size the parent of the other, breaking a tie arbitrarily.\  \textit{Linking by rank} maintains a non-negative integer \textit{rank} for each root, initially zero, and makes the root of larger rank the parent of the other, breaking\ a tie by adding one to the rank of one of the roots.  \textit{Linking by index} chooses a fixed total order of the nodes and makes the root of larger index the parent of the other.\  \textit{Linking by random index} is the special case of linking by index that chooses the total order of nodes uniformly at random.\par

\vspace{\baselineskip}
Linking by size, rank, or random index combined with naïve find, compression, splitting or halving gives an algorithm that takes O(log\textit{n}) time for an operation on a set or sets containing $n$\ elements, worst-case for deterministic linking, high-probability for linking by random index.  Use of compaction improves the amortized time per operation: any combination of compression, splitting, or halving with linking by size, rank, or random index gives an algorithm that takes $O(m \cdot \alpha(n, m/n))$ time to do $m$ operations on sets containing a total of \textit{n} elements.\  The bound is worst-case for linking by size or rank, average-case for linking by randomized index.\  Here  is a functional inverse of Ackermann’s\ function defined as follows.  We recursively define \textit{A\textsubscript{k}}(\textit{n}) for non-negative integers \textit{k} and \textit{n} as follows:\par

\vspace{\baselineskip}
\textit{A}\textsubscript{0}(\textit{n}) = \textit{n }+ 1; \textit{A\textsubscript{k}}(0) = A\textit{\textsubscript{k}}\textsubscript{–1}(1) if \textit{k} > 0; \textit{A\textsubscript{k}}(\textit{n}) = \textit{A\textsubscript{k}}\textsubscript{ – 1}(\textit{A\textsubscript{k}}(\textit{n} – 1)) if \textit{k} > 0 and \textit{n} > 0. \par

\vspace{\baselineskip}
For a non-negative integer \textit{n} and non-negative real-valued \textit{d}, 
$$\alpha(n,d) = \min\{k > 0 \mid A_k(\floor{d}) > n \}$$

\vspace{\baselineskip}
\begin{lemma} 
$A_k(n) < \min\{A_{k+1}(n), A_k(n+1)\}$, i.e., $A_k(n)$ is strictly increasing in $k$ and $n$.
\end{lemma}

\vspace{\baselineskip}
\begin{proof} 
The proof is by double induction on \textit{k} and \textit{n.}\  \textit{A}\textsubscript{0}(\textit{n}) = \textit{n }+ 1 < \textit{n }+ 2 =\textit{ A}\textsubscript{0}(\textit{n} + 1), and \textit{A}\textsubscript{0}(0) = 1 < 2 = \textit{A}\textsubscript{1}(0).\  Let\textit{ k} > 0.\  Suppose the lemma holds for \textit{k’} < \textit{k} and all \textit{n}.\  Then \textit{A\textsubscript{k}}(0) < \textit{A\textsubscript{k}}(0) + 1 = \textit{A}\textsubscript{0}(\textit{A\textsubscript{k}}(0))  \textit{A\textsubscript{k}}\textsubscript{–1}(\textit{A\textsubscript{k}}(0)) = \textit{A\textsubscript{k}}(1) = \textit{A\textsubscript{k}}\textsubscript{+1}(0).\  Thus the lemma also holds for \textit{k} and \textit{n }=\ 0.  Let \textit{k }> 0 and \textit{n}\ > 0.  Suppose the lemma holds for \textit{k’} < \textit{k} and all \textit{n}, and for \textit{k} and \textit{n}\ – 1.  Then \textit{A\textsubscript{k}}(\textit{n}) < \textit{A\textsubscript{k}}(\textit{n}) + 1 = \textit{A\textsubscript{0}}(\textit{A\textsubscript{k}}(\textit{n}))  \textit{A\textsubscript{k}}\textsubscript{–1}(\textit{A\textsubscript{k}}(\textit{n})) = \textit{A\textsubscript{k}}(\textit{n} + 1), and \textit{A\textsubscript{k}}(\textit{n}) = \textit{A\textsubscript{k}}(\textit{A\textsubscript{0}}(\textit{n} – 1)) < \textit{A\textsubscript{k}}(\textit{A\textsubscript{k}}\textsubscript{+1}(\textit{n} – 1)) = \textit{A\textsubscript{k}}\textsubscript{+1}(\textit{n}).
\end{proof}

\vspace{\baselineskip}
\begin{corollary}
$\alpha(n, d)$ is non-decreasing in $n$ and non-increasing in $d$. \par
\end{corollary}

\vspace{\baselineskip}
Our goal is to extend at least one sequential set union algorithm to the concurrent model of Section 2 and to obtain an almost-linear work bound that grows sublinearly with $p$, the number of processes.
For convenience in stating bounds, we assume that $2 \le p \le n \le m$, and that there is at least one unite of different elements.  
We denote the base-two logarithm by lg.\par

\section{Concurrent Linking and Splitting}

\vspace{\baselineskip}
Concurrency significantly complicates the implementation of the set operations.  One complication is that processes can interfere with each other by trying to update the same field at the same time, requiring our algorithms to be robust to such interference. Consider doing\ unites concurrently.  To do \textit{unite}(\textit{x},\textit{y}), we can start as in the sequential case by finding the roots \textit{u }and \textit{v} of the trees containing \textit{x} and \textit{y},\ respectively.  Then we can try to link \textit{u} and\textit{ v} by doing a CAS to make \textit{v} the parent of \textit{u} or vice-versa.\  But we must allow for the possibility that the CAS can fail, for example if it tries to make \textit{v }the parent of \textit{u} but in the meantime some other process makes another node the parent of \textit{u}.\  If this happens we must retry the unite.\  When retrying, we start the new finds at \textit{u} and \textit{v} rather than at \textit{x} and \textit{y}, to avoid revisiting nodes.\  Anderson and Woll \cite{AW91} proposed this method; the following pseudocode implements it.\  Method \textit{link}(\textit{u}, \textit{v}), to be defined, tries to make one of two roots \textit{u} and \textit{v }the parent of the other.\par

\vspace{\baselineskip}
\textit{unite}(\textit{x}, \textit{y}):\par

\ \ \  \textit{u }= \textit{find}(\textit{x}); \textit{v} = \textit{find}(\textit{y}); while $u \ne v^*$  do $ \{ $ \textit{link}(\textit{u}, \textit{v})$\ast$ ; \textit{u }= \textit{find}(\textit{u}); \textit{v} = \textit{find}(\textit{v})$ \} $ \par

\vspace{\baselineskip}
In\ this and subsequent implementations, asterisks denote linearization points.  The linearization point of a unite is the linearization point of the successful link if there is one, or the inequality test that fails if no link is successful.\par

\vspace{\baselineskip}
Concurrency also imposes constraints on the linking rule.\  We need to prevent concurrent links from creating a\ cycle of parent pointers other than a loop at a root.  For example, three concurrent links might make \textit{v} the parent of \textit{u}, \textit{w} the parent of \textit{v}, and \textit{u} the parent of \textit{w.}\  The simplest way to prevent such cycles is to do\ linking by index, which we can implement using CAS.  We denote the total order of nodes by $``$<$"$ .\  The following pseudocode implements linking by index:\par

\vspace{\baselineskip}
\textit{link}(\textit{u},\textit{ v}):\par

\ \ \  if \textit{u} < \textit{v} then CAS(\textit{u.p}, \textit{u},\textit{ v})${}^*$  else CAS(\textit{v.p}, \textit{v}, \textit{u})${}^*$ \par

\vspace{\baselineskip}
The\ linearization point of the link is its CAS.  A link is \textit{successful} if its CAS returns true.\  For\ any total order, linking by index guarantees acyclicity.  \textit{Linking by random index} is the special case of linking by index that chooses the total order uniformly at random.\par

\vspace{\baselineskip}
With this implementation of link, a link can succeed even though the new parent itself becomes a child of another node at the same time.\  Fortunately this affects neither correctness nor efficiency.\  We could prevent this anomaly by using DCAS to do links, which allows us to guarantee\ that\ the new parent remains a root.  But this has two drawbacks.  First, it uses DCAS, whereas our goal is to use only CAS if possible.\  Second, if all links are done using DCAS, the total work can be linear in \textit{p}, as we discuss in Section 5.1.\par

\vspace{\baselineskip}
Next we consider finds.\  Concurrent naïve finds do not interfere with each other, since such finds do not change the data structure.\  Thus\ we can do such finds exactly as in the sequential case.  The following pseudocode implements concurrent naïve find:\par

\vspace{\baselineskip}
\textit{find}(\textit{x}):\par

\ \ \  \textit{u} = \textit{x}; \textit{v} = \textit{u.p}${}^*$ ; while $v \ne u$ do \textit{u} = \textit{v}; \textit{v} = \textit{u.p}${}^*$ ; return \textit{u}\par

\vspace{\baselineskip}
The linearization point of a find is the last update of \textit{v}.\par

\vspace{\baselineskip}
Concurrent finds with compaction \textit{can} interfere with each other.\ \ Consider a sequential find with splitting.  Let \textit{u} be the current node visited by the find.  One step of the find consists of setting \textit{v} =\textit{ u.p}; setting \textit{w} =\textit{ v.p}; and, if $v \ne w$, replacing \textit{u.p }by \textit{w} and then setting \textit{u} = \textit{v}.\  Steps continue until \textit{v }=\textit{ w,} when the find finishes by returning \textit{v.}\  The only update to the data structure in a step is the replacement of \textit{u.p} by \textit{w.\  }We obtain a concurrent version of splitting by using CAS(\textit{u.p}, \textit{v}, \textit{w}) to do the\ update.  The following pseudocode implements this method, which we originally presented in \cite{JT16} and which is based on Anderson and Woll’s version of find with halving:\par

\vspace{\baselineskip}
\textit{find}(\textit{x}):\par

\ \ \  \textit{u} = \textit{x}; \textit{v =} \textit{u.p}; \textit{w} = \textit{v.p}${}^*$ ;\par

\ \ \  while $v \ne w$ do $ \{ $ CAS(\textit{u.p}, \textit{v}, \textit{w}); \textit{u} = \textit{v}; \textit{v} = \textit{u.p}; \textit{w} =\textit{ v.p}${}^*$ $ \} $ ;\par

\ \ \  return \textit{v}\ \    \  \par

\vspace{\baselineskip}
The linearization point of a find is the last update of \textit{w}.\  We call this method \textit{one-try splitting} because it tries once to update \textit{u.p} and then changes the current node from \textit{u} to \textit{v}, whether or not the update of \textit{u.p} has succeeded.\par

\vspace{\baselineskip}
Concurrent splits can produce anomalies that are not possible if splits are sequential, as a simple example shows.\  (See Figure 2.)\  Suppose \textit{a}, \textit{b}, \textit{c}, \textit{d},\textit{ e} is a path in a tree built by linking by index, and that four processes, 1, 2, 3, and 4 begin concurrent finds with one-try splitting starting at \textit{a,} \textit{a,} \textit{b}, and \textit{b},\ respectively.  We denote the local variables of process \textit{i} by \textit{u\textsubscript{i}}, \textit{v\textsubscript{i}, w\textsubscript{i}}.\  First, process 1 sets \textit{u}\textsubscript{1} = \textit{a}, \textit{v}\textsubscript{1} = \textit{a.p }= \textit{b}, and\textit{ w}\textsubscript{1}\textit{ =} \textit{b.p }= \textit{c}.\  Second, process 3 sets \textit{u}\textsubscript{3} = \textit{b}, \textit{v}\textsubscript{3} = \textit{c}, \textit{w}\textsubscript{3} = \textit{d}, and replaces \textit{b.p} by \textit{d.}\  Third, process 4 sets \textit{u}\textsubscript{4} = \textit{b}, \textit{v}\textsubscript{4} = \textit{d}, \textit{w}\textsubscript{4} = \textit{e}, and replaces \textit{b.p} by \textit{e.} Fourth, process 2 sets \textit{u}\textsubscript{2} = \textit{a}, \textit{v}\textsubscript{2} = \textit{b}, and \textit{w}\textsubscript{2 }= \textit{d}.\  Fifth, process 1 replaces \textit{a.p} by \textit{c}.\  Sixth, process 2 attempts to replace \textit{a.p} by \textit{d} but fails, because process 1 changed \textit{a.p} after process 2 read it.\  Observe that just before process 1 replaces \textit{a.p} by \textit{c}, \textit{c }is not an ancestor of \textit{a}, even though it was when process 1 read it.\ \ This threatens correctness.  Furthermore, even though the failure of process 2 to update \textit{a.p} guarantees that \textit{a.p} has changed since process 2 read it, the new value of \textit{a.p, }namely \textit{c}, is not an ancestor of the current grandparent of \textit{a}, namely \textit{e}, violating a property used in the analysis of sequential splitting.\  Finally, even though the new parent \textit{c} of \textit{a} is higher in index than the old parent \textit{b} of \textit{a} (as we prove in Theorem 1), the new grandparent \textit{d} of \textit{a} is \textit{lower} than the old grandparent \textit{e} of \textit{a}. \par

\begin{figure}[h]
    \centering
    \includegraphics[scale=0.35]{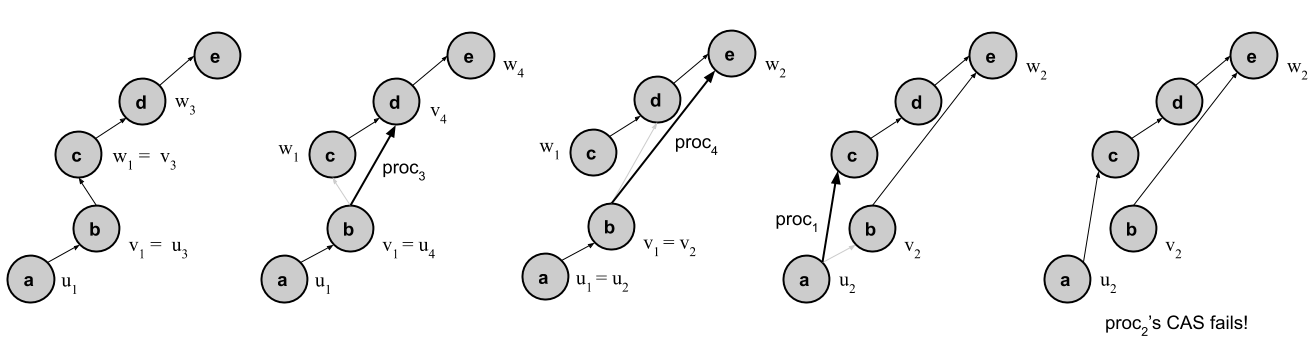}
    \caption{Interference in concurrent splitting: process 1 updates $a$'s parent to a node that is not its ancestor; process 2's CAS fails. These difficulties do not occur in the sequential setting.}
    \label{fig:2}
\end{figure}

\vspace{\baselineskip}
Fortunately, correctness requires only a weak property of compaction, one that holds for one-try splitting and many other methods.\  Suppose we do linking\ by index.  We define the \textit{union forest }to be the set of trees such that the parent of a node \textit{x} is the \textit{first }value of \textit{x.p} other than \textit{x}; if \textit{x.p }=\textit{ x} throughout the computation, \textit{x }is\ a root in the union forest.  We call a compaction method \textit{valid} if it visits nodes on a single path in the union forest, each vertex visit takes $O(1)$ steps, and each replacement of a parent \textit{y} by another node \textit{z} (of which there may be none)\textit{ }is such that \textit{z} is a proper ancestor of \textit{y} in the union forest. \  \par

\vspace{\baselineskip}
\begin{theorem}
Linking by index in combination with any valid compaction method maintains the invariant that the parents define a set of trees that partition the nodes into the correct disjoint sets, and the parent of any non-root node has higher index than the node.
Furthermore each set operation stops in $O(h)$ steps, where $h$ is the height (maximum number of edges on a path) of the union tree, so the algorithm is bounded wait-free.
\end{theorem}

\begin{proof}
An\ induction on the number of parent changes using the transitivity of $``$<$"$  shows that the parent of any node never has smaller index than the node.  This implies that the only cycles are loops at roots.  Parent changes done during compactions do not change the node partition defined\ by the trees.  A link that makes \textit{v} the parent of \textit{u }must be such that \textit{u} is a root before the link, \textit{u} < \textit{v}, and \textit{u} and\textit{ v} are in the trees containing the two nodes \textit{x} and\textit{ y }that\ are the inputs to the unite that does the link.  It follows by induction on the number of parent changes that at all times the trees correctly partition the nodes.\  Correctness of the linearization points follows in a straightforward way by induction on the number of parent changes.

\vspace{\baselineskip}
Since the nodes visited during a find are on a single path in the union forest, and each node visit takes $O(1)$ steps, each find stops in $O(h)$ steps. (Our assumption that there is at least one unite of different elements implies $h > 0$.)  The nodes visited during a unite are on two paths in the union forest.\  Consider the node\ visits in the order they occur.  Each node visit takes $O(1)$ steps, but a node can be visited many times.\ \ This can only happen while it is a root; once it becomes a child, it can only be visited once more (as the input to a find).  Consider the nodes \textit{u} and \textit{v} just before an execution of the test ``$u \ne v$" \ in unite.  Each of \textit{u} and \textit{v} was a root at some time\ during the find that computed it.  If the test ``$u \ne v$"  succeeds, whichever of \textit{u} and \textit{v} is smaller in the total order will be a child after the next link (whether or not the link succeeds).\  Suppose without loss of generality it is \textit{u}.\  We charge the next visits to \textit{u} and \textit{v} to \textit{u} becoming a child.\  There are at most $2h$ such events.\  It follows that the total number of node visits during the unite, and hence the total number of steps, is $O(h)$.
\end{proof}

\vspace{\baselineskip}
Having dealt with correctness, we discuss concurrent compaction in more detail.\  The monotonicity of parents (each new parent is higher in index than the old one) allows us to extend the analysis of sequential splitting to one-try splitting, although the extension is not straightforward.\  On the other hand, the analysis of sequential halving relies on monotonicity of grandparents, which fails in the concurrent setting, as our example above shows.\  Anderson and Woll \cite{AW91} claimed a good work bound for their concurrent version of halving, but they overlooked the problem of non-monotonicity.\  We see no way to get a good work bound for their method.\par

\vspace{\baselineskip}
Even though we can prove good efficiency bounds for one-try splitting, we can prove slightly better bounds for a related compaction method that tries to change each parent pointer twice instead of once.\  We call this method \textit{two-try splitting}.\  The following pseudocode implements find with two-try splitting:\par

\vspace{\baselineskip}
\textit{find}(\textit{x}):\par

\ \ \  \textit{u} = \textit{x}; \textit{v =} \textit{u.p}; \textit{w} = \textit{v.p}${}^*$ ;\par

\ \ \  while $v \ne w$ do $ \{ $ CAS(\textit{u.p}, \textit{v}, \textit{w}); \textit{v} = \textit{u.p}; \textit{w} =\textit{ v.p};\textit{ }CAS(\textit{u.p}, \textit{v}, \textit{w}); \textit{u} = \textit{v}; \textit{v} = \textit{u.p}; \textit{w} =\textit{ v.p}${}^*$ $ \} $ ;\par

\ \ \  return \textit{v}\ \ \ \ \ \  \par

\vspace{\baselineskip}
The linearization point of a find is the last assignment to \textit{w}.\  If every attempted parent change succeeds, the effect of a single two-try split is to replace the parent of every other node on the find path by its great-grandparent.\  This splits the original path into two paths, each containing half the nodes on the original path, but the split is different from that produced by one-try splitting: if the nodes on the original path are numbered consecutively from 1, the latter produces a path of nodes 1, 3, 5, 7$ \ldots $  and another path of nodes 2, 4, 6, 8$ \ldots $ ; the former produces a path of nodes 1, 4, 5, 8, 9$ \ldots $  and another path of nodes 2, 3, 6, 7, 10, 11$ \ldots $ \  \par

\vspace{\baselineskip}
A variant that has the same work bounds as two-try splitting is \textit{conditional two-try splitting}, in which the second try occurs only if the first one fails.\  We omit a detailed discussion of this variant, since its pseudocode is a bit longer and it is unclear whether avoiding extra parent changes improves efficiency.\par

\vspace{\baselineskip}
Both one-try and two-try splitting are valid compaction methods, so Theorem 1 holds for both of them. \par

\vspace{\baselineskip}
We conclude this section by presenting Anderson and Woll’s concurrent implementation of \textit{same-set}, which gives an extension of our algorithms\ to their formulation of the problem.  It is easy to do \textit{same-set}(\textit{x}, \textit{y}) in the sequential setting: find the root \textit{u} of the tree containing \textit{x}, find the root \textit{v} of the tree containing \textit{y}, and test whether \textit{u} = \textit{v.\  }As Anderson and Woll observed, this does not suffice in the concurrent setting, because \textit{u}\ might\ no longer be a root when the equality test occurs, possibly invalidating the test.  Their solution has three cases.  If \textit{u} = \textit{v}, return true: \textit{x} and \textit{y} are in the same tree when the test occurs, and remain in the same tree.\  If $u \ne v$, test whether \textit{u} is still a root.\  If so, return false: \textit{x} and \textit{y} were in different trees when \textit{v} was computed, since \textit{u} and \textit{v }were different roots.\  If not, redo the computation: do new finds from \textit{u} and\textit{ v}, and repeat the test or tests.\  The following pseudocode implements this method:\par

\vspace{\baselineskip}
\textit{same-set}(\textit{x}, \textit{y}):\par

\ \ \  \textit{u }= \textit{find}(\textit{x}); \textit{v} = \textit{find}(\textit{y})${}^*$ ;\par

\ \ \  while $u \ne v$ do $ \{ $ w = \textit{u.p}; if \textit{u} = \textit{w} then return false; \textit{u }= \textit{find}(\textit{u}); \textit{v} = \textit{find}(\textit{v})${}^*$ $ \} $ ;\par

\ \ \  return true\par

\vspace{\baselineskip}
The linearization point of a same-set is the last assignment to \textit{v}.\  All our analyses of find and unite extend to include \textit{same-set} as an allowed operation.\par

\section{Concurrent Linking by Rank}

\vspace{\baselineskip}
To obtain a good work bound, we combine\ one-try or two-try splitting with a good linking method.  Linking by random index is one such method, but our analysis of it assumes that the scheduling of CAS instructions is independent of the random node order.\  This assumption is questionable; if it fails, the work bound becomes much worse as a function of $p$, as we show in Section 8.\  To overcome this, we develop two concurrent versions of linking by rank, one deterministic and one randomized, both of which have good work bounds.\  To simplify our descriptions, we assume for the moment that the rank and parent of a node can be stored in a single block of memory that is updatable by one\ CAS instruction.  In Section 6 we show how to eliminate this assumption.\par

\vspace{\baselineskip}
Both\ of our versions of linking by rank are refinements of a generic method.  The generic method links roots of different ranks using CAS, and links roots of the same rank using method \textit{elink}, to be defined. The rank of node \textit{u} is \textit{u.r}, initially zero.\  The following pseudocode implements the generic method: \par

\vspace{\baselineskip}
\textit{link}(\textit{u}, \textit{v}):\par

\textit{\ \ \  r} =\textit{ u.r}; \textit{s} = \textit{v.r};\par

\ \ \  if \textit{r} < \textit{s} then CAS((\textit{u.p}, \textit{u.r}), (\textit{u}, \textit{r}), (\textit{v}, \textit{r}))${}^*$ \par

\ \ \  else if \textit{r }> \textit{s }then CAS((\textit{v.p}, \textit{v.r}), (\textit{v}, \textit{s}), (\textit{u}, \textit{s}))${}^*$ \par

\ \ \  else \textit{elink}(\textit{u}, \textit{v}, \textit{r})$\ast$ ;\par

\vspace{\baselineskip}
Given two roots \textit{u }and \textit{v} with ranks \textit{r} and \textit{s}, respectively, this method compares \textit{r} to \textit{s.\  If r }< \textit{s, }it uses a CAS to make \textit{v }the parent of \textit{u} while guaranteeing that neither the parent nor the rank of \textit{u }changes\ in the meantime.  If \textit{r }>\textit{ s},\ it proceeds symmetrically.  If \textit{r =} \textit{s,} it does an elink to link \textit{u }and \textit{v}.  A link is successful if its CAS returns true or its \textit{elink} is successful, in which case the linearization point of the link is its CAS or that of its \textit{elink}.\  Our two versions of linking by rank differ only in their implementation of \textit{elink.}\par

\subsection{Linking by Rank via DCAS}

\vspace{\baselineskip}
A simple way to do \textit{elink}(\textit{u}, \textit{v, r})  is to use a DCAS to make \textit{v} the parent of \textit{u} and increment the rank of \textit{v} while guaranteeing that the ranks and parents of \textit{u} and \textit{v} do not change in the meantime.. The following pseudocode implements this idea:\  \par

\vspace{\baselineskip}
\textit{elink}(\textit{u}, \textit{v}, \textit{r}):\par

\ \ \  DCAS((\textit{u.p}, \textit{u.r}), (\textit{u}, \textit{r}), (\textit{v}, \textit{r}), (\textit{v.p}, \textit{v.r}), (\textit{v, r}), (\textit{v}, \textit{r} + 1))$\ast$ ;\par

\vspace{\baselineskip}
An elink is successful if its DCAS returns true, in which case the linearization point of the elink is its DCAS.\ \  \par

\vspace{\baselineskip}
Our first version of linking by rank uses this implementation of elink.\  The rank of a node can never decrease, and\ can increase only while the node is a root.  It follows that the rank of a child is always strictly less than that of its parent\textit{.}\  Linking by rank is an implicit form of linking by index: the successful links respect any total order consistent with the final ranks of nodes.  Thus Theorem 1 holds for this method.\par

\vspace{\baselineskip}
The following lemma and theorem extend known bounds on sequential linking by rank \cite{TvL84} to linking by rank via DCAS:\par

\vspace{\baselineskip}
\begin{lemma}
With linking by rank via DCAS, the sum of ranks is at most $n-1$, the number of nodes of rank $k$ is at most $n/2^k$, and the maximum rank and the height of the union forest are at most $\lg n$.
\end{lemma}

\begin{proof}
For a node to increase in rank by 1, it must be a root, and another root must become its child at the same time.\  It follows that the number of rank increments, and hence the sum of ranks, is at most \textit{n}\ – 1, one per root that becomes a child.  An induction on \textit{k} shows that at most \textit{n}/2\textit{\textsuperscript{k}} nodes can ever attain rank \textit{k}.\  The bounds on the maximum rank and the height of the union forest follow, since no node can have rank exceeding lg\textit{n}.
\end{proof}

\vspace{\baselineskip}
\begin{theorem}
Linking by rank via DCAS in combination with any valid compaction method maintains the invariant that the parents define a set of trees that partition the nodes into the correct disjoint sets, and the rank of a child is less than that of its parent.\  Furthermore each set operation stops in $O(\log n)$ steps, so the algorithm is bounded wait-free.
\end{theorem}

\begin{proof}
The first half of the theorem follows by induction on the number of steps as in the proof of the first half of Theorem 1.\  A find takes $O(\log n)$ steps by the argument in the proof of Theorem 1, since the height of the union forest is $O(\log n)$ by Lemma 2.\  We prove the bound for unites by an extension of the argument in the proof of Theorem 1.\  The nodes visited during a unite are on two paths in the union forest, and on each path they are visited in increasing\ order by rank.  Each node visit takes O(1) steps, but roots can be visited many times.\  We charge each repeated visit to a root either to a root becoming a child or to a root increasing in rank.\  Consider the nodes \textit{u} and \textit{v} just before an execution of the test ``$u \ne v$" \ in unite.  Each of \textit{u} and \textit{v}\ was a root at some time during the find that computed it.  Suppose the test  ``$u \ne v$" \ succeeds.  The next execution of link sets \textit{r} to the rank of \textit{u} and\textit{ s} to the rank of \textit{v}.\  If \textit{r }< \textit{s}, then after the CAS either the rank of \textit{u} has increased or \textit{u }has\ become a child, whether or not the CAS succeeds.  We charge the next visits to \textit{u} and \textit{v} to the rank increase of \textit{u} or to \textit{u}\ becoming a child.  The symmetric argument applies if \textit{r} > \textit{s}.\  If \textit{r} =\textit{ s}, at least one of \textit{u} and\textit{ v} has increased in rank or become a child after the elink.\  We charge the next visits to \textit{u} and \textit{v}\ to whichever of these events has occurred.  There are at most $2\lg n$ roots that become children and at most $2\lg n$ rank increases by Lemma 2, since for each of the two paths in the union forest the rank increases sum to at most $\lg n$.\  It follows that the total number of node visits during the unite, and hence the total number of steps, is $O(\log n)$.
\end{proof}

\vspace{\baselineskip}
The efficiency of this linking\ method (though not its correctness) depends critically on using CAS to link nodes of different ranks, reserving DCAS for the equal-rank case.  An attempted link of equal-rank nodes \textit{u }and \textit{v} using DCAS fails only if some other process makes \textit{u} or \textit{v} a non-root, or increases the rank of \textit{u} or \textit{v}.\  In the proof of Theorem 2 we charge extra node visits resulting from\ the failure of the DCAS to whichever of these events occurs.  If we were to use DCAS to try to make a node \textit{v} the parent of a node \textit{u} of lower rank, the DCAS could fail because another process made \textit{v} the parent of another node \textit{w}.\textit{\  }This changes neither the parent nor rank of \textit{u}, nor of \textit{v}, leaving us with no event to charge for extra node visits.\  In the worst case, $O(n)$ such links could produce $\Omega(pn)$ failures, resulting in total work linear in $p$.\ \ Using CAS to link nodes of different ranks eliminates these failures.  Although we can avoid such interference in the disjoint set union problem as we have defined it, this is much harder to do in some extensions of the problem, as we discuss in Section 9.\par

\subsection{Randomized Linking by Rank}

\ \   \  \par

To\ link\ equal-rank nodes using CAS, we need to do the parent change and the rank increment separately.  The question is which one to do first.  Making\ this decision randomly gives an approximation to linking by rank that produces few enough rank ties that we are able to get good work bounds.  Since this method allows rank ties, we use linking by index to break such ties, in order to prevent the creation of non-trivial\ cycles of parent pointers.  Assume that $``$<$"$  is an arbitrary total order\ of the nodes.  To link two equal-rank roots \textit{u} and \textit{v} such that \textit{u} <\textit{ v},\ we flip a fair coin.  If it comes up heads, we attempt to make \textit{v} the parent of \textit{u}; if it comes up tails, we attempt to increase the rank of \textit{u}.\  The following pseudocode implements this idea. Random Boolean method \textit{flip} returns true with probability $\sfrac{1}{2}$ and false otherwise, independent of all other flips.\par

\vspace{\baselineskip}
\textit{elink}(\textit{u}, \textit{v},\textit{ r}):\par

\ \ \   if \textit{u} < \textit{v} then\par

\ \ \ \ \ \ \  $ \{ $ if \textit{flip} then CAS((\textit{u.p}, \textit{u.r}), (\textit{u}, \textit{r}), (\textit{v}, \textit{r}))$\ast$  else CAS((\textit{u.p}, \textit{u.r}), (\textit{u}, \textit{r}), (\textit{u}, \textit{r} + 1))$ \} $ \par

\ \ \  else\par

\ \ \ \ \ \ \  $ \{ $ if \textit{flip} then CAS((\textit{v.p}, \textit{v.r}), (\textit{v}, \textit{r}), (\textit{u}, \textit{r}))$\ast$  else CAS((\textit{v.p}, \textit{v.r}), (\textit{v}, \textit{r}), (\textit{v}, \textit{r} + 1))$ \} $ \par

\vspace{\baselineskip}
An elink is successful if it does a CAS that changes a parent pointer, in which case the linearization point of the elink is its CAS.\  \par

\vspace{\baselineskip}
Our second version of linking by rank uses this implementation of elink.\  Observe that the CAS done after a flip is almost the same whether the flip returns true or false, the only difference being the updated field (parent or rank, respectively).\  In our analysis we shall assume that the success or failure of the CAS following a flip is independent of the outcome\ of the flip.  In Section 6 we describe how to modify the implementation to eliminate the need for this independence assumption.\  Randomized linking by rank is an implicit form of linking by index: links respect the total order defined by final node ranks with ties broken by $``$<$"$ .\par

\vspace{\baselineskip}
\begin{lemma}
With randomized linking by rank, 
\emph{(i)} any node \textit{x} has $O(1)$ ancestors of the same rank, in expectation; 
\emph{(ii)} the sum of ranks is at most $n$ in expectation and $n + O(n^{\sfrac{1}{2}})$ with probability $1 - 1/n^c$ for any constant $c > 0$, where the constant factor in the $``$O$"$  depends on $c$; 
\emph{(iii)} the expected number of nodes of rank at least $k$ is at most $n/2^k$, and with probability at least $1 - n/2^k$, all nodes have rank less than $k$; 
\emph{(iv)} the maximum rank is at most $\lg n + 3$ in expectation and is at most $(c + 1)\lg n$ with probability at least $1 - 1/n^c$ for any positive constant $c$; 
\emph{(v)} the depth of the union forest is at most $3\lg n + 9$ in expectation and $O(\log n)$ with probability at least $1 - 1/n^c$ for any constant $c > 0$, where the constant factor inside the $``$O$"$  depends on $c$; and 
\emph{(vi)} for a large enough constant $c$ and any $k > 0$, the expected number of nodes of rank less than $k$ and height at least $ck$ in the union forest is at most $n/2^k$.
\end{lemma}

\begin{proof}
Consider only flips that result in successful CAS operations.  Each such flip produces a rank increment with probability $\frac{1}{2}$; otherwise, it makes a root into a child.\par

\vspace{\baselineskip}
(i) The probability that a node has \textit{k }ancestors of the same rank is at most 1/2\textit{\textsuperscript{k}}\textsuperscript{–1}.\  Summing gives the bound.\par

\vspace{\baselineskip}
(ii) There are at most \textit{n }– 1 flips that make a root\ into a child.  The sum of ranks, which is the number of rank increments, is thus at most the number of heads in a sequence of coin flips containing at most \textit{n }tails, which is at most \textit{n }in expectation, and at most \textit{n }+ O(\textit{n}\textsuperscript{1/2}) with probability 1 – \textit{n}\textsuperscript{c} for any constant \textit{c} > 0\textit{ }by a Chernoff bound \cite{CL05}, with the constant factor in the $``$O$"$  depending on \textit{c}.\par

\vspace{\baselineskip}
(iii) The rank of a given node is at least \textit{k} with probability at most 1/2\textit{\textsuperscript{k}}.\  The expected number of nodes of rank at least \textit{k} is thus at most \textit{n}/2\textit{\textsuperscript{k}}.\  By a union bound, all nodes have rank less than \textit{k} with probability at least 1 – \textit{n}/2\textit{\textsuperscript{k}}.\ \  \par

\vspace{\baselineskip}
(iv) For \textit{c }> 0, the probability that the maximum rank is at least lg\textit{n }+ \textit{c} is at most \textit{n}/2\textsuperscript{lg\textit{n} + \textit{c}} = 1/2\textit{\textsuperscript{c}}.\  It follows that the expected maximum rank is at most lg\textit{n} +  \(  \sum _{i=1}^{\infty}i/2^{i} \)   lg\textit{n} + 2  lg\textit{n} + 3, and the probability that the maximum rank exceeds (\textit{c} + 1)lg\textit{n} is at most 1/2\textit{\textsuperscript{c}}\textsuperscript{lg\textit{n}} = 1/\textit{n\textsuperscript{c}}.\par

\vspace{\baselineskip}
(v)\ A node gains a proper ancestor of the same rank with probability at most $\sfrac{1}{2}$.  Thus the expected number of proper ancestors of the same rank as that of a given node is at most  \(  \sum _{i=1}^{\infty}i/2^{i} \) \textbf{ = }2, which implies by (ii) that the expected depth of the union forest is at most 3lg\textit{n}\ +\ 9.  Let c > 0.  By (ii) the maximum node rank is at most (\textit{c }+ 3)lg\textit{n} with probability at least 1 – 1/\textit{n\textsuperscript{c}}\textsuperscript{+2}.\  For any node, the probability that it has at least (\textit{b} + \textit{c} + 3)lg\textit{n} proper ancestors in the union forest is at most the probability that a sequence of fair coin flips containing at most $(c + 3)\lg n$ heads contains at least \textit{b}lg\textit{n}\ tails.  By a Chernoff bound, for \textit{b} sufficiently large, this probability is at most 1 – 1/\textit{n\textsuperscript{c}}\textsuperscript{+2}.\  The probability that at least one of the \textit{n} nodes has more than (\textit{b} + \textit{c} + 3)lg\textit{n }proper ancestors in the union forest is at most 
$2n/n^{c+2} \le 1/n^c$.

\vspace{\baselineskip}
(vi) Let \textit{x}\ be any node.  We claim that for some \textit{c} >\textit{ }0 the probability that \textit{x} has an ancestor \textit{y }of rank less than \textit{k} such that the path from \textit{x} to \textit{y} contains \textit{ck} edges is at most 1/2\textit{\textsuperscript{k}}.\  Part (vi)\ follows from the claim by a union bound.  To prove the claim, consider the edges on the path from \textit{x }to \textit{y}\ in the union forest.  Call such an edge\textit{ good} if its ends have different ranks and \textit{bad} otherwise.\  Each edge has probability\ at least $\sfrac{1}{2}$ of being good, independent of the status of all other edges.  The claim follows by a Chernoff bound.
\end{proof}

\begin{theorem}
Randomized linking by rank in combination with any valid compaction method maintains the invariant that the parents define a set of trees that partition the nodes into the\ correct disjoint sets.  The parent of any non-root node has rank\ no less than that of the node, and if the ranks are equal, the parent has larger index.  Each set operation stops in $O(\log n)$ steps with probability $1 - 1/n^c$ for any $c > 0$, where the constant factor in the $``$O$"$  depends on $c$, so the algorithm is bounded wait-free with probability $1 - 1/n^c$ for any $c > 0$.
\end{theorem}

\begin{proof}
The theorem follows from parts (iv) and (v) of Lemma 3 by a proof like those of Theorems 1 and 2.
\end{proof}

\begin{remark}
We can modify randomized linking by rank to guarantee that it is bounded wait-free without affecting any of our bounds by placing an upper bound of $n-1$ on\ the maximum rank: when linking two nodes of maximum rank, we do not do a flip but always try to make the smaller-index node a child of the larger.
\end{remark}

\subsection{Linking by Random Index}

\vspace{\baselineskip}
With an appropriate definition of rank, Lemma 3 and Theorem 3 hold for linking by random index, under a strong independence assumption.\  We define the rank of node \textit{x} to be lg\textit{n}\ –  lg(\textit{n} –\textit{ x}\  + 1).\  Thus node \textit{n} has rank lg\textit{n}, nodes \textit{n} – 1 and \textit{n} – 2 have rank lg\textit{n}\ \ – 1, and so on.  The rank of a child is no greater than that of its parent.\  We use these ranks only in the analysis; the implementation of the algorithm does not use them.\par

\vspace{\baselineskip}
We assume that the random node order is independent of the linearization of the unite operations.\ \ More precisely, we assume that the node order and linearization are generated together in the following way.  The implementation maintains a set \textit{U} of unordered pairs $\{u,v\}$  that are candidates for linking, initially empty, and a partial order \textit{P} of the nodes, initially empty, that is a total order on the nodes of any set defined by the links done so far, and that leaves any two nodes in different sets unordered.\  To do \textit{link}(\textit{u,} \textit{v}), a process adds the unordered pair $\{u,v\}$ to \textit{U}.\par

\vspace{\baselineskip}
The scheduler sequentially removes pairs from \textit{U }in\ arbitrary order.  When removing a pair $\{u,v\}$  from \textit{U},\ the scheduler performs three actions.  First, it modifies \textit{P} by merging the total orders of the sets containing \textit{u} and \textit{v}, with each possible merged order equally likely.\  Second, if \textit{u} < \textit{v} it sets \textit{u.p} = \textit{v}, if \textit{v} < \textit{u} it sets \textit{v.p} = \textit{u}.  This unites the sets containing \textit{u} and \textit{v}.\  The link corresponding to the pair $\{u,v\}$ \ succeeds.  Third, the scheduler deletes from \textit{U} all other pairs containing \textit{u} or \textit{v}.\  Each link corresponding to such\ a pair fails.  When a pair is deleted from \textit{U}, the process that added the pair to \textit{U} proceeds with its next operations, which are the recomputing of its \textit{u} and \textit{v}.\par

\vspace{\baselineskip}
The updating of \textit{P} maintains the invariant that the total order of the nodes in any set\ defined by the links done so far is uniformly random.  If there is more than one set after all unites have been done, we can extend the final partial order to a total order by merging the total orders on the final sets,\ with each possible merged order equally likely.  The result is a uniformly random permutation of the nodes, equivalent to a uniformly random numbering of the nodes from 1 to \textit{n}.\  Thus this implementation does linking by random index, subject to the restriction imposed on the scheduler.\  The execution does not change if we initially number the nodes uniformly at random but reveal to the scheduler only the total order within each set formed so far.\par

\vspace{\baselineskip}
This implementation restricts the behavior of linking by random index in at least two different ways: in the actual implementation, the CAS operation for a link is defined by an ordered pair, not an unordered pair, so the scheduler gets information about the node order before it needs to decide which such CAS to do next.\  Also, in the actual implementation a link can make a root the child of a non-root, which cannot happen with the scheduler constraint.\  We think the latter restriction\ is\ inconsequential, but the former is significant.  We thus view our analysis of linking by random index as suggestive, not definitive.  \ \    \ \  \par

\vspace{\baselineskip}
\begin{lemma}
With linking by random index, if the scheduler restriction holds, then parts \emph{(i)}, \emph{(v)}, and \emph{(vi)} of Lemma 3 are true, as well as the following strengthened versions of parts \emph{(ii)}, \emph{(iii)}, and \emph{(iv)}: \emph{(ii)} the sum of ranks is at most $n$, \emph{(iii)} the number of nodes of rank at least $k$ is at most $n/2^k$, and \emph{(iv)} the maximum node rank is at most $\lg n$.
\end{lemma}

\begin{proof}
Parts (ii), (iii), and (iv) are immediate from the definition of ranks.\  Parts (i), (v), and (vi) follow as in the proof of Lemma 3 from the following claim:
\end{proof}

\vspace{\baselineskip}
(*) Given a node \textit{x}, each successive ancestor of \textit{x} in the union forest has probability at least ½ of having higher rank then its parent, independent for each ancestor.\par

\vspace{\baselineskip}
To prove (*), consider a node \textit{x}.\  Given an execution of the algorithm, we modify the linearization by delaying links uniting sets not containing \textit{x} until such a set is a subset of a set about to be united with \textit{x}.\  That is, let \textit{S} be the current set containing \textit{x}, let $\{u,v\}$  be the next pair deleted from \textit{U} with \textit{u} but not \textit{v} in\textit{ S}, and let \textit{S’} be the current set containing \textit{v}.\  We modify the schedule to delay all the links forming \textit{S’} until just before the link of \textit{u} and \textit{v}.\ \ This does not change the steps done by the execution, only their linearization order.  We generate the partial order \textit{P}\ by generating a numbering of the nodes incrementally and revealing to the scheduler only the total order within each set constructed so far.  Initially we assign a number uniformly at random to \textit{x}.\  Subsequently when the scheduler removes a pair $\{u,v\}$  from \textit{U}, if \textit{u} and/or \textit{v} is unnumbered we assign it a number chosen uniformly at random from the numbers not yet assigned.\par

\vspace{\baselineskip}
Suppose the scheduler removes a pair $\{u,v\}$  from \textit{P} with \textit{u} the root of the tree containing \textit{x}, and \textit{v}\ in another tree, and that the corresponding link succeeds.  Let \textit{S} and \textit{S’}, respectively, be the sets containing \textit{u} and \textit{v} just before $\{u,v\}$  is removed from \textit{U}.\  Just before the links forming \textit{S’} are done, the only numbered nodes are those in \textit{S}, and \textit{u}\ has the largest number.  Among the numbers larger than that of \textit{u}, at least half have rank larger than that of \textit{u}.\  When \textit{S’}\ is formed, its nodes are numbered uniformly at random from among the unassigned numbers.  Given that a number assigned to a node in \textit{S’} is larger than that of \textit{u}, it has probability at least $\sfrac{1}{2}$ of being larger than the rank of \textit{u}.\  Since \textit{v} has largest number among the nodes in \textit{S’}, if the link makes \textit{v} the parent of \textit{u} the rank of \textit{v} is greater than that of \textit{u}\ with\ probability\ at least $\sfrac{1}{2}$.  The claim (*) follows by induction on the number of steps.   \par

\begin{theorem}
Theorem 3 holds for linking by random index if the scheduler restriction holds.\par
\end{theorem}

\begin{proof}
The theorem follows from parts (iv) and (v) of Lemma 4 in the same way that Theorem 3 follows from parts (iv) and (v) of Lemma 3.
\end{proof}

\section{Indirection and Helping}

\vspace{\baselineskip}
The algorithms in Sections 5.1 and 5.2 require that CAS (and DCAS in 5.1) support testing and updating of storage\ blocks able to store both the parent and the rank of a node.  In this section we present two ways to eliminate the need for blocks to contain multiple fields.\  \  Both methods increase the number of steps per operation, but by at most a constant factor.\  We also discuss how to modify the implementation of the randomized linking algorithm of Section 5.2 to eliminate the need for an unrealistic independence assumption in its analysis.\par

\vspace{\baselineskip}
The first method to reduce the block size, proposed by Anderson and Woll \cite{AW91} is to use \textit{indirection}.\  Specifically, each node contains only one field, a pointer to a \textit{ledger} that contains the parent and rank of the node.\ \ To\ do\ a link via a CAS, a process creates a new ledger containing the updated information for the node being linked and then uses a CAS to attempt to replace the old ledger of the node by the new one.  A link via a DCAS is similar, except that the process creates two new ledgers and uses a DCAS to replace the old ledgers of the two affected nodes.  Parent updates done by splitting are done directly on the appropriate ledgers, without allocating new ones.  The ledger method requires a way to allocate ledgers, and care must be taken to avoid\ the reuse of ledgers.  The algorithm of Section 5.1 needs at most $3n - 2 + 2p = O(n)$ ledgers, one per initial set plus at most two per successful link\ plus at most two per process.  The algorithm of Section 5.2 needs $O(n)$ ledgers with high probability.\  If ledgers are used to implement randomized linking by rank, the independence assumption needed by the analysis becomes much weaker and quite realistic: the success or failure of a CAS can depend on all inputs to the CAS, in particular the ledger addresses, but not on the contents of the ledgers.\par

\vspace{\baselineskip}
Allocating ledgers efficiently is itself a challenging problem, which Anderson and Woll ignored.  
One way to do it is to use the concurrent fetch and increment method of Ellen and Woelfel \cite{EW}.
If ledgers are allocated individually, the number of steps to allocate a ledger is $O(\log p)$.  
If ledgers are allocated in groups of $O(\log p)$, the amortized time per allocation is $O(1)$ and the number of ledgers used will be $O(n)$ provided $p = O(n/\log n)$.\  The problem of ledger allocation deserves more serious study.\  \ \  \par

\vspace{\baselineskip}
The second method is \textit{helping}, as described for example in \cite{HFP}.\  The idea is to allow processes to complete the tasks of other processes.\  We number the processes from 1 to $p$.\  Each node \textit{u} has an extra field, \textit{u.process},\ which can hold a process number or 0, and is initially 0.  Each process has a \textit{descriptor} in which it records a sequence of steps it wants to perform. To update a node, a process writes appropriate instructions into its descriptor and then does a CAS or DCAS to write its process number into the \textit{process}\ field\ of\ the affected node or nodes.  Any other process that wants to update a node containing a non-zero process number must first execute the instructions in the corresponding descriptor.  When the last instruction is executed, the process number in the affected node or nodes is reset to 0, allowing further updates to the node or nodes.  As long as the number of instructions needed to do an update is bounded by a constant, the use of descriptors increases the total work by only a constant factor.\par

\vspace{\baselineskip}
In using helping to link nodes of equal rank, we have to solve the \textit{ABA problem}: A helping process, having completed the instructions in a descriptor, resets the process number in the relevant node to zero, but in the meantime the process being helped has reset the process number to zero, initiated a new update, and set the process number to its own number again.\  The\ helping\ process\ has no way to detect that a new update has been initiated by the process that initiated the old one.  In our application, we can solve the ABA problem by using the monotonicity of ranks.  This requires that a CAS be able to update the rank and the process number of a node as an atomic operation.  Since ranks are small and in any realistic application $p \ll n$, we think this is a reasonable assumption.\par

\vspace{\baselineskip}
The\ deterministic algorithm of Section 5.1 does links using helping as follows.  Each descriptor contains two nodes and\ a rank.  To link root \textit{x} of rank \textit{r} to root \textit{y}, a process, say process \textit{i}, writes \textit{x}, \textit{y}, and \textit{r} into its descriptor.\  If \textit{y} has rank greater than \textit{r}, it uses a CAS to write \textit{i} into node \textit{x} while verifying that the process number of \textit{x} is 0 and the rank of \textit{x} is \textit{r}.\  If \textit{y} has rank \textit{r}, it uses a DCAS to write \textit{i} into both \textit{x} and \textit{y} while verifying that the process numbers of \textit{x} and \textit{y} are 0 and the ranks of \textit{x} and \textit{y} are \textit{r.}\  A process wanting to update a node that finds a non-zero process number \textit{i} in the node reads the corresponding descriptor.\  Suppose the descriptor contains nodes \textit{x} and \textit{y} and rank \textit{r}.\  The process sets \textit{z} = \textit{y.p} and does a CAS to set the parent of \textit{x} to \textit{z} while verifying that \textit{x}\ was a root before the update.  It then tests whether \textit{y} is a root of rank \textit{r}.\  If so, it does a CAS to change the rank of \textit{y} to \textit{r }+ 1 and the process number of \textit{y} to 0 while verifying that the rank and process number of \textit{y} were \textit{r} and \textit{i} before the update.\par

\vspace{\baselineskip}
This method uses a couple of optimizations.  It does not reset the process number of a node that becomes a non-root, since no subsequent link will try to change its parent or rank.  Instead of making \textit{y} the parent of \textit{x}, it makes \textit{y.p} the parent of \textit{x}.  
The reason to do the link this way is that some other process can make \textit{y} a non-root just before process \textit{i} does its CAS or DCAS to write \textit{i} into $x$, or into $x$ and $y$.
If this happens, $y.p$ will have rank greater than $r$ when $x$ becomes its child, preserving the invariant that ranks strictly increase from child to parent.
If this does not happen and the rank of \textit{y} is \textit{r}, 
the helping process adds one to the rank of $y$ and resets the process number of $y$ to 0.\par

\vspace{\baselineskip}
The randomized algorithm of Section 5.2 does helping using descriptors containing a node \textit{y},\textit{ }a rank \textit{r}, and a \textit{flag}\ whose value is null, true, or false.  A flag of true indicates that \textit{y} should become the parent of the root containing the process number of the descriptor; a flag of false indicates that the rank of this node should be changed from \textit{r} to \textit{r}\ + 1.  If process \textit{i} wants to link root \textit{x} of rank \textit{r} to root \textit{y}, it writes \textit{y} and \textit{r}\ into its descriptor.  If \textit{y} has rank greater than \textit{r}, it sets the flag to true; if the rank of \textit{y} equals \textit{r},\ it flips a fair coin and sets the flag correspondingly.  Then it uses a CAS to set the process number of \textit{x} equal to \textit{i} while verifying that the rank and process number of \textit{x} were \textit{r} and 0 before the update.\  A process wanting to update a node \textit{x} that finds a non-zero process number \textit{i} in \textit{x}\ reads the corresponding descriptor.  If the flag is true it does a CAS to set the parent of \textit{x} to \textit{y} while verifying that \textit{x}\ was a root before the update.  If the flag is false it does a CAS to set the rank and process number of \textit{x} to \textit{r} + 1 and 0 while verifying that they were \textit{r} and \textit{i} before the update.\par

\vspace{\baselineskip}
The\ analysis of this method relies on the same independence assumption as the method using ledgers: scheduling decisions are independent of the contents of descriptors.  A variant of the method is to set the flag \textit{after} the process number of the descriptor is written into \textit{x}:\ the\ first step of a helping process is to change the flag from null to true or false using the randomized CAS operation mentioned in Section 2.  If this operation is available, no independence assumption is needed.  This implementation of randomized linking by rank satisfies the Anderson-Woll requirement that a randomized algorithm be efficient even if the scheduler knows the outcome of previous random choices.\ \ We think, though, that it is reasonable to assume that the scheduler makes its decisions only on the basis of the inputs to the CAS operations, or that it cannot read the private memories of the processes.  If either of these assumptions hold, we do not need randomized CAS.\par

\section{Upper Bounds}

\vspace{\baselineskip}
The results of Sections 5 and 6 give us the following theorem:\par

\begin{theorem}
With any of the three linking methods of Section 5 combined with any valid compaction method, the total work is $O(m \log n)$.\  This bound is worst-case for the deterministic linking method, high-probability for the randomized methods.\  If randomized linking by rank is implemented as described in Section 6, the bound is valid even for an adversarial scheduler.
\end{theorem}

\begin{proof}
The theorem is immediate from the results of Sections 5 and 6. 
\end{proof}

\vspace{\baselineskip}
The use of splitting instead of naïve find improves the total work bounds significantly if $p \ll n$. We show this by extending the analysis of sequential splitting \cite{TvL84, GKLT14} to one-try and two-try splitting.\par

\vspace{\baselineskip}
We define the \textit{density} $d$ of a set union problem instance to be $m/(np)$ if splitting is two-try, $m/(np^2)$ if splitting is one-try.  We shall obtain a bound of $O(m \cdot (\alpha(n, \textit{d}) + \log (1 + 1/\textit{d})))$ on the total work if either kind of splitting is used in combination with any of the three linking methods.\  The main obstacle we encounter in extending the sequential analysis to the concurrent setting is accounting for unsuccessful CAS operations.\  Accounting for such operations adds the logarithmic term to the work bound.\par

\vspace{\baselineskip}
We call a problem instance \textit{sparse} if \textit{d} < 1 and \textit{dense} otherwise.\  The logarithmic term in the work bound dominates only in sparse instances.\  We start with the analysis of dense instances, which is simpler than that of sparse ones.\par

\vspace{\baselineskip}
We call a child a \textit{zero child}\ if its rank is the same as that of its parent.  Zero children only exist if a randomized linking method is used.\par

\vspace{\baselineskip}
With deterministic linking by rank or linking by random index, the maximum rank is lg\textit{n},\ but with randomized linking by rank, the maximum rank is unbounded, although large ranks occur with exponentially small probability (Lemma 3 part (iii)).  Our first result handles the high-rank case.\par

\vspace{\baselineskip}
\begin{lemma}
With randomized linking by rank, the probability that the maximum rank is \textit{n} or greater is at most $1 - n/2^n$.\  The contribution to the expected total work of instances in which the maximum rank is $n$ or greater is $O(m)$.
\end{lemma}

\begin{proof}
Part\ (iii) of Lemma 3 gives the first part of the lemma.  If the maximum rank is \textit{k}, the number of steps done during any find or unite operation is $O(n + k)$.\  The contribution to the expected running time of instances in which the maximum rank is $n$ or greater is thus bounded by a constant times \textit{m \(  \sum _{k \geq n}^{} \left( n+k \right) /2^{k} \) } = $O(m)$.
\end{proof}

\vspace{\baselineskip}
Henceforth when considering randomized linking by rank we assume all ranks are less than $n$.\ \ \ \  \par

\subsection{The Dense Case}

\vspace{\baselineskip}
Throughout this section we assume $d \ge 1$. \par

\vspace{\baselineskip}
\begin{lemma}
The number of finds is $O(m)$, worst-case unless randomized linking by rank is used, in which case the bound is with high probability.
\end{lemma}

\begin{proof}
There are at most two finds per unite plus at most two per process per root that increases in rank or becomes a child, for a total of $O(m + np) = O(m)$.
For randomized linking, this bound follows from part (ii) of Lemmas 3 and 4 and is high-probability for randomized linking by rank, worst case for linking by random index.
\end{proof}

\vspace{\baselineskip}
We call a node \textit{low} if its rank is less than $d$ and \textit{high} otherwise.\  The number of nodes of rank at least $n$ is zero if linking is deterministic or by randomized index, and is zero with probability at least $1 - n/2^n$ by if linking is randomized by rank.
During a find, a {\em visit} to a node is an iteration of the find loop in which the node is the value of $u$.
(See the pseudocode in Section 4.)

\vspace{\baselineskip}
\begin{lemma}
The number of visits to low nodes during finds is $O(m)$, worst-case if linking is deterministic, expected if randomized.
\end{lemma}

\begin{proof}
Consider three successive visits to low nodes during a find, to \textit{u}, \textit{v}, and \textit{w}.\  Let \textit{I} be the interval of time between the visits of \textit{u} and \textit{w}.\  We claim that at least one of the following events occurs during \textit{I}: \textit{u} or \textit{v} becomes a child, \textit{u.p.r} increases, or \textit{u} or \textit{v}\ loses an ancestor of the same rank.  The number of such events for fixed \textit{u} and \textit{v} is $O(d)$: a node only becomes a child once, its parental rank can increase at most \textit{d} times before it exceeds $d$ and its parent is not low; a node has $O(1)$ ancestors of the same rank in expectation by part (i) of Lemmas 3 and 4.\  We charge the visit to \textit{u}\ to the corresponding event (or any such event if there is more than one).  Each event is charged for at most 2\textit{p}\ visits,\ at most two per process.  (The factor of two comes from the two nodes associated with a visit, the node itself and the next node visited.)  Summing over all nodes, we obtain a bound of $O(npd) = O(m)$ on visits to low nodes.

Suppose the claim is false.  
The $u$ and $v$ are children when $u$ is visited.
After the CAS following the visit to \textit{u}, the parent of \textit{u} has changed; after the CAS following the visit to \textit{v}, the parent of \textit{v}\ has changed.  If either \textit{u} or \textit{v }is a zero child when \textit{u} is visited, at least one of them becomes a non-zero child or loses an ancestor of the same rank during \textit{I}.\  Thus neither \textit{u} nor \textit{v} is a zero child when \textit{u}\ is visited.  But then the rank of the parent of \textit{u} increases by the time the CAS after the visit to \textit{u}\ finishes,
making the claim true.
\end{proof}

\vspace{\baselineskip}
Bounding visits to high nodes is more complicated.\  For each high child \textit{x}, we measure the progress of compaction by keeping track of an increasing function of the rank of the parent of \textit{x}, called the \textit{count} of \textit{x}.\  We define counts using Ackermann’s\ function.  Our formulation is an extension of that of Kozen \cite{Kozen}.\  We define the \textit{level x.a} of a high node \textit{x}, and the\textit{ index x.b} and \textit{count x.c} of a high child \textit{x,} as follows:\par

\vspace{\baselineskip}
\textit{x.a} = min$ \{$\textit{k}$ \vert $ \textit{A\textsubscript{k}}(\textit{x.r}) > \textit{x.p.r}$ \} $ ;\par

\textit{x.b} = max$ \{$\textit{i}$ \vert $ \textit{A\textsubscript{x.a}}(\textit{i}) $ \leq $  \textit{x.p.r}$ \} $ ;\par

\textit{x.c} = \textit{x.r  \(  \cdot   \) x.a} + \textit{x.b}.\par

 \[  \cdot  \] \par

We bound the range of levels, indices, and counts by using the properties of Ackermann’s function:\par

\vspace{\baselineskip}
\begin{lemma}
If \textit{x} is a high node, $0 \leq x.a \leq \alpha(n, d)$ and $x.a = 0$ if and only if \textit{x.r }=\textit{ x.p.r}.\  If \textit{x} is a high child, $0 \le$  \textit{x.b }< \textit{x.r} and $0 \leq x.c < (\alpha(n,d) + 1) x.r$.  
The values of \textit{x.a} and \textit{x.c} never decrease, and if \textit{x.a} or \textit{x.b }increases, \textit{x.c} increases by at least as much.
\end{lemma}

\begin{proof}
Since \textit{A}\textsubscript{0}(\textit{x.r}) = \textit{x.r }+ 1, \textit{x.a} = 0 if and only if \textit{x.r} = \textit{x.p.r}, and $x.a \ge 0$ if it is defined.
If \textit{x} is a high node, $A_{\alpha(n,d)}(x.r) \ge A_{\alpha(n,d)}(\floor{d}) > n > x.p.r$. 
Thus \textit{x.a} is defined and is at most $\alpha(n,d)$.
(Here for randomized linking by rank we use the assumption that all ranks are less than \textit{n}.)  
Suppose \textit{x}\ is a high child.  If \textit{x.a }= 0, \textit{x.b} = \textit{x.r }– 1 since $x.r \ge d \ge 1$.
If  \textit{x.a }> 0, \textit{A\textsubscript{x.a}}(0) = \textit{A\textsubscript{x.a}}\textsubscript{–1}(1) $ \leq $  \textit{A\textsubscript{x.a}}\textsubscript{–1}(\textit{x.r}) $ \leq $  \textit{x.p.r}, so \textit{x.b} is defined. 
Since \textit{A\textsubscript{x.a}}(\textit{x.r}) > \textit{x.p.r}, \textit{x.b }< \textit{x.r}.\  The bounds on \textit{x.c} follow from those on \textit{x.a} and \textit{x.b}.\  While \textit{x} is a root, \textit{x.a }=\ 0.  Once \textit{x} is a child, \textit{x.r} is constant and \textit{x.p.r }cannot decrease, so \textit{x.a}\ cannot decrease by Lemma 1.  While \textit{x.a }is constant, \textit{x.b }cannot\ decrease for the same reason.  If \textit{x}.\textit{a} increases by one, \textit{x.b} can decrease by at most \textit{x.r }– 1, resulting in an increase of at least one in \textit{x.c}.\  If \textit{x.a} increases by more than \textit{k}, \textit{x.c} increases by at least (\textit{k} – 1)\textit{x.r} + 1.
\end{proof}

\vspace{\baselineskip}
\begin{lemma}
The sum of the counts of all high children is $O(n \alpha(n,d))$, worst-case unless linking is randomized by rank, in which case the bound is high-probability.
\end{lemma}

\begin{proof}
By Lemma 8, the sum of the counts of high children is $O(\alpha(n,d))$ times the sum of the ranks of all nodes.  By Lemmas 2 and 4, the sum of ranks is less than $n$ for deterministic linking by rank and linking by random index.  
For randomized linking by rank, it is $O(n)$ with high probability by Lemma 3.
\end{proof}

\vspace{\baselineskip}
The following lemma is the key to the analysis of splitting.\par

\vspace{\baselineskip}
\begin{lemma}
Consider a time \textit{t} at which \textit{u} is a high child whose parent \textit{v} is also a $($high$)$ child. Let \textit{w} the parent of \textit{v} at time \textit{t}, and let \textit{u.a, v.a}, and \textit{w.r} be the levels of \textit{u} and\textit{ v }and the rank of \textit{w} at time \textit{t},\ respectively.  Suppose that at time \textit{t} or later the parent of\textit{ u} changes from \textit{v} to a node \textit{x} of rank at least \textit{w.r}.\  If \textit{v.a} > \textit{u.a}, the parent change increases \textit{u.a} and \textit{u.c }by at least \textit{v.a} – \textit{u.a}; if \textit{v.a }=\textit{ u.a}, the parent change increases \textit{u.c} by at least 1 or causes \textit{u} to lose an ancestor of the same rank.
\end{lemma}

\begin{proof}
Let \textit{u.r} and \textit{v.r }be the ranks of \textit{u} and \textit{v} at time \textit{t}, respectively.  Let \textit{x.r} be the rank of \textit{x} when it becomes the parent of \textit{u}.\  Since \textit{A\textsubscript{v.a}}\textsubscript{–1}(\textit{u.r}) < \textit{A\textsubscript{v.a}}\textsubscript{–1}(\textit{v.r}) $ \leq $  \textit{w.r} $ \leq $  \textit{x.r}, the level of \textit{u} after the parent change is at least \textit{v.a}.\  If \textit{v.a > u.a}, the parent change increases the level and hence the count of \textit{x} by at least \textit{v.a }– \textit{u.a} by Lemma 8.\  Suppose \textit{v.a }= \textit{u.a}.\  If \textit{u.a }= 0, the parent change causes \textit{u} to lose \textit{v}\ as an ancestor.  Suppose \textit{u.a} > 0. \  Since \textit{A\textsubscript{u.a}}(\textit{u.b} + 1) =  \textit{A\textsubscript{u.a}}\textsubscript{–1}(\textit{A\textsubscript{u.a}}(\textit{u.i})) $ \leq $  \textit{A\textsubscript{u.a}}\textsubscript{–1}(\textit{v.r}) $ \leq $  \textit{w.r} $ \leq $  \textit{x.r}, the parent change increases either the level or the index of \textit{u} and hence increases the count of \textit{u}.
\end{proof}

\vspace{\baselineskip}
To count visits to high nodes, we use a credit argument.\  One credit pays for one high-node\ visit.  We allocate a certain number of credits to each find when it starts, and additional credits when high nodes increase in count or lose ancestors of the same rank.\  We show via a \textit{credit invariant} that these credits suffice to pay for all the high-node visits.\  A bound on the total number of credits gives a bound on the number of high-node visits.\par

\vspace{\baselineskip}
We begin by analyzing two-try splitting: even though it is more complicated than one-try splitting, its analysis is simpler.\  We call a find \textit{active} while it is being executed.\  When a find starts, we allocate it \textit{$ \alpha $ }(\textit{n}, \textit{d}) + 1\ credits.  When the count of a high child increases by \textit{k}, we allocate 2\textit{k} credits to each active find, for a total of at most 2\textit{pk}.\  When a high child loses an ancestor of the same rank, we allocate one credit to each active find, for a total of at most \textit{p}.\par

\vspace{\baselineskip}
\begin{lemma}
With two-try splitting, the number of allocated credits is $O(m \alpha(n,d))$, worst-case if linking is deterministic, average-case if randomized.
\end{lemma}

\begin{proof}
By Lemma 6, the number of credits allocated to finds when they start is $O(m \alpha(n,d))$.  
By Lemma 9, the number of credits allocated to finds as a result of increases in count is $O(np \alpha(n,d)) = O(m \alpha(n,d))$.  
By Lemmas 3 and 4, the expected number of credits allocated to finds as a result of nodes losing ancestors of the same rank is $O(np) = O(m)$.
\end{proof}

\vspace{\baselineskip}
\begin{lemma}
With two-try splitting, just after a high node \textit{u} is visited by a find, the find has at least \textit{u.a} credits.
\end{lemma}

\begin{proof}
We prove the lemma by induction on the number of high-node visits done by a find.\  When the find starts, it has $O(\alpha(n,d) + 1$ credits.\  The first visit costs one, leaving $\alpha(n,d)$,\ which is enough to make the lemma true just after this visit.  Suppose the lemma holds just after \textit{u} is visited, and let \textit{v} be the next node visited.\  We denote by unprimed and primed values their values just after the visit to \textit{u} and just before the visit to \textit{v},\ respectively.  The lemma holds after the visit to \textit{v} provided that the find accrues at least \textit{v.a}’ – \textit{u.a} + 1 credits between the visits to \textit{u} and \textit{v}.\  To show that this happens, we need the following crucial inequality, which follows from Lemma 10:\par

\vspace{\baselineskip}
(*) $u.a' \ge v.a$

\vspace{\baselineskip}
To\ prove (*), we refer to the implementation of two-try splitting.  Let \textit{t} be the first time \textit{u.p} = \textit{v}.\  Time \textit{t} is after the visit to \textit{u}, since \textit{u.p} changes between the first and second times that the find sets its variable \textit{v }after the visit to \textit{u}, as a result of the first CAS after the visit to \textit{u} succeeding or failing.\  Let \textit{w} be the parent of \textit{v} at time \textit{t}.\  Consider the change to \textit{u.p} resulting from the second CAS after the visit to \textit{u}.\  This change satisfies the hypothesis of Lemma 9, since the new parent of \textit{u} must have been the parent of \textit{v} at time \textit{t}\ or later.  By Lemma 10, just after this change to \textit{u.p}, the level of \textit{u} is at least the level of \textit{v} at time \textit{t}.\  Since levels are non-decreasing, (*) holds.\par

\vspace{\baselineskip}
Between the visits to \textit{u} and \textit{v}, the find accrues at least 2(\textit{u.a}’ – \textit{u.a} + \textit{v.a}’ – \textit{v.a}) = (\textit{v.a}’ – \textit{u.a}) + (\textit{u.a}’ – \textit{u.a}) + (\textit{v.a}’ – \textit{v.a}) + (\textit{u.a}’ – \textit{v.a}) credits as a result of level increases..\ \ Each of the last three terms is non-negative, the last one by (*).  Thus the find accrues at least \textit{v.a}’ – \textit{u.a} + 1 credits between the visits, unless the levels of \textit{u} and \textit{v} are equal and unchanging\ between the visits.  Suppose the levels of \textit{u} and\textit{ }v\textit{ }are equal and unchanging between the visits.\  By Lemma 10, the find accrues at least one credit when the parent of \textit{u} changes from \textit{v}.
\end{proof}

\vspace{\baselineskip}
\begin{lemma}
With two-try splitting, the number of visits to high nodes is $O(m \alpha(n,d))$, worst-case if linking is deterministic, average-case if randomized.
\end{lemma}

\begin{proof}
The lemma is immediate from Lemmas 11 and 12.
\end{proof}

\vspace{\baselineskip}
Now we extend the analysis to one-try splitting.\  The proof of Lemma 12 fails for one-try splitting, because a CAS done by one process, say process 1, can fail as a result of a successful CAS done by another process, say process 2, that sets its value of \textit{v} \textit{before }process\ 1’s most recent high-node visit.  That is, time \textit{t} in the proof of Lemma 10\ can precede the visit.  This invalidates the use of Lemma 10 in the proof.\par

\vspace{\baselineskip}
To overcome this problem, we allocate additional credits to node count increases, and we allow active finds to shift some of their credits to the other active finds.\  Specifically, when a find starts, we allocate it $\alpha(n,d) + 1$ \textit{normal} credits.\  When a high child loses an ancestor of the same rank, we allocate one normal credit to each active find.\  When the count of a high node increases by $k$, we allocate 2\textit{k }normal credits and $2k(p-1)$ \textit{extra} credits to each active find.\  When a CAS in a find succeeds, we shift a $1/(p-1)$ fraction of the find’s extra credits to each other active find.  Shifted extra credits become normal; that is, we shift a credit at most once.\ \   \par

\vspace{\baselineskip}
\begin{lemma}
With one-try splitting, the number of allocated credits is $O(m \alpha(n,d))$, worst-case if linking is deterministic, average-case if randomized.
\end{lemma}

\begin{proof}
The bound holds for normal credits by the proof of Lemma 11.\  By Lemma 9, the number of extra credits allocated to finds as a result of increases in count is $O(np^2\alpha(n,d)) = O(m\alpha(n,d))$ since $d = m/(np^2)$.
\end{proof}

\vspace{\baselineskip}
\begin{lemma}
With one-try splitting, just after a high node \textit{u} is visited by a find, the find has at least \textit{u.a} normal credits.
\end{lemma}

\begin{proof}
The proof is an extension of that of Lemma 12.\  Consider a find, say find 1.\  The credits allocated to the find when it starts make the lemma true just after its first high-node visit.\  Suppose the lemma holds just after find 1 visits \textit{u}, and let \textit{v} be the next node it visits.\  We consider three cases.\  If the CAS after the visit of find 1 to \textit{u} succeeds, the lemma holds just after the visit to \textit{v} by an argument like that in the proof of Lemma 12.\ \ (This case does not use shifted credits.)  Suppose this CAS fails, because a CAS done by another find, say find 2, changes \textit{u.p} from \textit{v}\ to another value.  Let \textit{t} be the last time that find 2 set its variable \textit{v}\ before its successful CAS.  If\textit{ t} is after find 1 visits \textit{u}, the lemma holds just after the visit to \textit{v} by an argument like that in the proof of Lemma 12, again without the use of shifted credits.\par

\vspace{\baselineskip}
The third, new case is if \textit{t} precedes the visit of find 1 to \textit{u}.\  Let \textit{t}’, \textit{t}’’, and \textit{t}’’’\textit{ }be the times find 1 visits \textit{u}, find 2 does its CAS, and find 1 visits \textit{v},\ respectively.  We denote by unprimed, primed, double-primed, and triple-primed values their values at times \textit{t}, \textit{t}’\textit{,} \textit{t}’’, and \textit{t}’’’, respectively.\  Applying Lemma 9 to time \textit{t} and the successful CAS of find 2 gives \textit{u.a}’’ $\ge$  v.a; and, if $u.a \le v.a$, the count of \textit{u} increases\textit{ }by at least 1 or \textit{u} loses an ancestor of the same rank when find 2 does its CAS. \par

\vspace{\baselineskip}
At time \textit{t’}, find 1 has at least \textit{u.a}’\textit{ }normal\ credits by the induction hypothesis.
Between times \textit{t’} and \textit{t}’’’, it accrues at least 2(\textit{u.a}’’’\textit{ – u.a}’ + \textit{v.a}’’’ – \textit{v.a}’)\ normal credits.  
Between times \textit{t} and \textit{t}’’, find 2 accrues at least 2(\textit{p} – 1)(\textit{u.a}’’\textit{ – u.a} + \textit{v.a}’’ – \textit{v.a}) $\ge$  2(\textit{p} – 1)(\textit{u.a}’\textit{ – u.a} + \textit{v.a}’ – \textit{v.a}) extra credits, of which at least 2(\textit{u.a}’\textit{ – u.a} + \textit{v.a}’ – \textit{v.a}) are shifted to find 1 and become normal at time \textit{t}’’: find 1 is active at \textit{t}’’\ since its CAS fails as a result of the CAS by find 2 succeeding.  
Thus between \textit{t}’\textit{ }and \textit{t}’’’ find 1 accrues at least 2(\textit{u.a}’’’\textit{ – u.a} + \textit{v.a}’’’ – \textit{v.a}) $\ge$  (\textit{v.a}’’’ – \textit{u.a}’) + (\textit{u.a}’’’ – \textit{u.a}) + (\textit{v.a}’’’ – \textit{v.a}) + (\textit{u.a}’’ – \textit{v.a})\ normal credits.  
Since \textit{u.a}’’ $\ge$ \textit{v.a}, this is at least \textit{v.a}’’’ – \textit{u.a}’ + 1, enough to make the lemma true for the visit to \textit{v}, unless \textit{u} and \textit{v} have equal and unchanging levels from \textit{t} to \textit{t}’’’, in which case find 1 accrues a normal credit when find 2 does its CAS.
\end{proof}

\vspace{\baselineskip}
\begin{lemma}
With one-try splitting, the number of visits to high nodes is $O(m \alpha(n,d))$,
worst-case if linking is deterministic, average if randomized.
\end{lemma}

\begin{proof}
The lemma is immediate from Lemmas 14 and 15.
\end{proof}

\subsection{The Sparse Case}

\vspace{\baselineskip}
In\ this section we modify the analysis of Section 7.1 to handle sparse instances.  Throughout this section we assume \textit{d}\ < 1.  We need to change the definition of low and high nodes, add an additional node type, \textit{middle}, and (for the purpose of the analysis only) redefine the ranks of nodes.\par

\vspace{\baselineskip}
Let \textit{l} = lg(1 + 1/\textit{d}).\  Since \textit{d} < 1, \textit{l}\ > 1.  A node is \textit{low} if its rank is less than \textit{l} and its height is less than \textit{cl}, where \textit{c }is the constant in part (vi) of Lemmas 3 and 4; \textit{middle} if its rank is less than \textit{l} but its height is at least \textit{cl}, and \textit{high} if its rank is at least \textit{l}.\ \ Middle nodes can exist only if linking is randomized.  \par

\vspace{\baselineskip}
\begin{lemma}
The number of non-low nodes is at most $2nd$, as is the sum of the ranks of such nodes.  This bound is worst-case if linking is deterministic, average-case if randomized.
\end{lemma}

\begin{proof}
By part (vi) of Lemmas 3 and 4, the expected number of middle nodes is at most \textit{n}/2\textit{\textsuperscript{l}} $ \leq $  \textit{n}/2\textsuperscript{lg(1/\textit{d})} = \textit{nd}\ if linking is randomized.  (It is zero if not.). By Lemma 2 or part (iii) of Lemma 3 or 4 depending on the linking method, the number of high nodes is also at most \textit{n}/2\textit{\textsuperscript{l}} $ \leq $  \textit{nd}, worst-case if linking is deterministic or by randomized index, average-case if by randomized rank.\  The bound on the sum of ranks follows from the node bound by the argument in the proof of Lemma 2 if linking is deterministic, by that in the proof of part (ii) of Lemma 3 or 4\ if\ randomized.
\end{proof}

\vspace{\baselineskip}
\begin{lemma}
The number of finds that visit at least one non-low node is $O(m)$, worst-case if linking is deterministic, average-case if randomized.
\end{lemma}

\begin{proof}
Consider\ the finds during unites that visit at least one non-low node.  At most two per unite also visit a low node.  
Of those that visit only non-low nodes, there are at most two per unite plus at most $2p$ per non-low node that becomes a child or has a rank increase, two per process doing a unite while the event in question takes place.  By Lemma 17, the number of such finds is $O(npd) = O(m)$.
\end{proof}

\vspace{\baselineskip}
\begin{lemma}
The number of visits to low nodes is $O(ml)$, worst-case.
\end{lemma}

\begin{proof}
The analysis of node visits in the proof of Theorem 2 restricted to nodes of rank less than $l$ and height less than $cl$ gives a bound of $O(l)$ low-node\ visits\ for each find and unite.
\end{proof}

\vspace{\baselineskip}
\begin{lemma}
If linking is randomized, the expected number of visits to middle nodes is $O(ml)$.
\end{lemma}

\vspace{\baselineskip}
\begin{proof}
By Lemma 18, the expected number of finds that visit middle nodes is $O(m)$. 
During such a find, each visit to a middle node except the last two is followed by a middle node losing a child of the same rank or the parent of a middle node \textit{x} increasing in rank.  
The latter can only happen $l$ times before \textit{x} has a parent that is not a middle node; subsequently, \textit{x}\ can only be the last middle node visited during a find.  We charge each visit to a middle node other than the last two of a find\ to the corresponding event.  The charge per event is at most $p$, and the expected number of events is at most $ndl$ rank increases and $O(nd)$ losses of same-rank ancestors, the latter by part (i) of Lemma 3 or 4.\  Such events account for $O(npdl) = O(ml)$ visits.
Adding the last two per find gives the lemma.
\end{proof}

\vspace{\baselineskip}
To count visits to high nodes, we define the \textit{effective rank} of a high node \textit{x} to be \textit{x.er} = \textit{x.r }– \textit{l}\ +\ 1.\ \ We define levels of high nodes and indexes and counts of high children, using effective ranks in place of ranks.  Since the effective rank of a high node is at least one, levels of high nodes and indices and counts of high children are well-defined.  We allocate credits exactly as in Section 7.1.  Lemmas 8 and 10\ remain true.  By Lemma 17, the sum of counts of high children is $O(nd\alpha(n,d))$, worst-case\ if linking is deterministic, high-probability if randomized.  We allocate credits exactly as in Section 5.1.  Lemmas 12 and 15\ remain true.  If splitting is two-try, the number of allocated credits is $O((m + ndp)\alpha(n,d)) = O(m \alpha(n,d))$ since $d = m/(np)$; if splitting is one-try, it is $O((m + ndp^2)\alpha(n,d)) = O(m \alpha(n,d))$ since $d = m/(np^2)$.  
We conclude that Lemmas 11, 13, 14, and 16\ hold\ in\ the\ sparse\ case\ (with\ the\ new\ definition\ of\ a\ high\ node).\par

\subsection{The Total Work Bound}

\vspace{\baselineskip}
Combining the results of Sections 7.1 and 7.2, we obtain the following theorem:\par

\vspace{\baselineskip}
\begin{theorem}
With any of the three linking methods of Section 5 and either one-try or two-try splitting, the total work is $O(m(\alpha(n,d) + \log(1 + 1/d)))$, worst-case if linking is deterministic, average-case if randomized, where $d = m/(np^2)$ if splitting is one-try, $d = m/(np)$ if splitting is two-try.
\end{theorem}

\begin{proof}
The theorem follows from Lemmas 5, 7, 13, 16, 19, and 20.
\end{proof}

\section{Lower Bounds}

\vspace{\baselineskip}
In this section, we derive lower bounds on the worst-case and amortized efficiency of set union algorithms.
In the first subsection, we prove lower bounds on the work efficiency of the algorithms described in this paper by explicitly providing worst-case executions---both the operations and the adversarial schedules.
At a high level, our executions are constructed by the following observations and steps.
For each algorithm, we describe operations that build a tree of logarithmic height using $\textit{unite}$ operations.
We observe that {\em shadowing schedules} in which all processes are scheduled in lock-step while performing the same expensive $\textit{find}$ operations result in worst-case behavior.
We apply a shadowing schedule to processes performing a $\textit{find}$ on the deepest node in the aforementioned tree to prove that the logarithmic term in our upper bounds is tight.
Then, we combine the idea of shadowing schedules with previous sequential lower bounds of Tarjan et al. and Fredman et al. \cite{TvL84, FS89} to show that the inverse-Ackermann term in our upper bounds is tight.
Our algorithmic lower bounds section proves that our amortized upper bound analyses are tight when find operations are done with two-try splitting.\par

\vspace{\baselineskip}
In the second subsection, we show general lower bounds that apply to the concurrent set union problem.
First, we prove that, in the worst-case, any concurrent set-union algorithm must do at least $\Omega(\log \min \{n, p\})$ work in expectation for a single operation.
When $p = n^{\omega({\frac{1}{\log\log n}})}$, this lower bound is stronger than the sequential lower bound of $\Omega\left(\frac{\log n}{\log\log n}\right)$ given by Fredman and Saks \cite{FS89} in the cell probe model.
It also shows a separation in work complexity between the sequential and concurrent versions of the set-union
problem, since Blum \cite{Blum} presented an algorithm that does at most $O\left(\frac{\log n}{\log\log n}\right)$ work per operation in the sequential setting.
Furthermore, whenever $\log p = \Theta(\log n)$, i.e. when $p = n^\epsilon$, this lower bound establishes that randomized linking with any form of compaction yields an algorithm with optimal expected work per operation.
Finally, we generalize the worst-case lower bound using shadowing schedules to show that our algorithm obtained by combining randomized linking with two-try splitting is optimal amongst a class of {\em symmetric algorithms} that includes all known algorithms for the concurrent disjoint set union problem.\par

\subsection{Algorithmic Lower Bounds}

\vspace{\baselineskip}
In order to prove the tightness of the inverse-Ackermann term in our upper bounds,
we recall a sequential cell probe lower bound on the set union problem given by Fredman and Saks.\par

\vspace{\baselineskip}
\begin{lemma}[\cite{FS89}]
    Let $\mathcal{A}$ be any randomized algorithm that solves the sequential set union problem.
    For any fixed number of nodes $n$, and any $M \ge n$, there is a sequence of operations $\sigma_M$, 
    that makes $\mathcal{A}$ perform $\Omega(M \alpha(n, M/n))$ expected work.
	\label{lem:sequentialLowerBound}
\end{lemma}

\vspace{\baselineskip}
We use Lemma \ref{lem:sequentialLowerBound} to establish a concurrent lower bound.

\vspace{\baselineskip}
\begin{lemma}
    Let $\mathcal{A}$ be any of the algorithms we have described for concurrent set-union.
    There is some sequence of $m$ operations using $p$ processes on $n$ nodes
    that requires 
    $\Omega\left(m \cdot \alpha{\left(n, \frac{m}{np} \right)} \right)$ work in expectation.
    
\label{lem:concurrentAckermannBound}
\end{lemma}

\begin{proof}
    Any concurrent algorithm is also a sequential algorithm if it is run by a single process.
    So, for any given $M \ge n$, we can take a worst-case sequence $\sigma_M$ of operations from
    Lemma \ref{lem:sequentialLowerBound}.
    That is, a single process running the sequence of operations $\sigma_M$ will perform $\Omega(M \alpha(n, M/n))$ work in expectation.
	In the remainder of the proof, we use shadowing schedules, in which processes run in lock-step with each other
	and thereby do not gain locally from any compaction attempts of other processes, to get the lower bound.\par

    \vspace{\baselineskip}
	We consider two cases for $m$:
	\begin{itemize}
		\item[]
		{\em Case 1:} 
		If $m \ge np$, then we choose $M = m/p \ge n$.
		If each of the $p$ processes runs $\sigma_M$ and is scheduled in lock-step (so that the processes all walk up find sequences together and do not benefit from each other's compaction attempts),
		then the total number of operations is $pM = m$ and the total amount of work is 
		$\Omega(pM\alpha(n,M/n)) = \Omega\left(m \alpha{\left(n, \frac{m}{np} \right)} \right)$.
		
		\item[]
		{\em Case 2:} 
		If $m < np$, we choose $M = n$.
		Then, $\sigma_M$ performed by a single process takes $\Omega(n \alpha(n, 1))$ expected work.
		We observe that $m/n < p$ and assign $m/n$ processes the operation sequence $\sigma_M$, thus assigning $m$ operations.
		If the processes are scheduled in lock-step, the total expected work performed by them is 
		$\Omega(m/n \cdot n \alpha(n,1)) = \Omega\left(m \alpha{\left(n,\frac{m}{np}\right)}\right)$.
	\end{itemize}
\end{proof}

\vspace{\baselineskip}
We can also show that the logarithmic term $\log(\frac{np}{m} + 1)$ is an amortized lower bound for all our algorithms.
The schedule that builds binomial trees with a single process
and makes all the processes shadow each other up the longest
branch of these trees yields the lower bound.\par

\vspace{\baselineskip}
\begin{lemma}
    For randomized linking by rank and linking by DCAS, regardless of what type of compaction is used in \textit{find} operations,
	and for any positive integer $k \in [1,n]$, 
	there is a sequence of $k - 1$ \textit{unite} operations that will build a tree with $k$ nodes with height 
	$\Omega(\log k)$.
    \label{lem:treeBuilder}
\end{lemma}

\begin{proof}
    For simplicity, we initially assume that $k$ is a power of 2.
    Let $B_j$ be the binomial tree of height $j$.
    All the nodes are initially in singleton trees, i.e. $B_0$ trees.
    We proceed in $\lg k$ rounds.
    In round $r$, we start with $\frac{k}{2^{r-1}}$ trees of type $B_{r-1}$,
    and simply unite their roots pairwise.
    After $\lg k$ rounds we end up with a single tree of type $B_{\lg k}$.
    If $k$ were not a power of 2, we perform the above procedure with the largest power of 2 less than $k$ and simply
    unite the remaining nodes to the root of the main tree.
\end{proof}

\vspace{\baselineskip}
A slightly more complex construction allows us to prove a similar lemma for linking by index.

\vspace{\baselineskip}
\begin{lemma}
    For randomized linking by index, regardless of what type of compaction is used in \textit{find} operations,
	and for any positive integer $k \in [1,n]$, 
	there is a sequence of $k - 1$ \textit{unite} operations by a single process that will build a tree with $k$ nodes in which the depth of a uniformly randomly picked node is 
	$\Omega(\log k)$ in expectation.
    \label{lem:treeBuilder2}
\end{lemma}

\begin{proof}
The proof is constructive. 
The construction of these trees is inspired by {\em binomial trees}, and is done in multiple {\em rounds} such that 
each round fully finishes before the next round starts.
Without loss of generality let $k$ be a power of 2, as otherwise we could just use the greatest power of 2 less than $k$ in the following construction.
Initially, we let the nodes be in singleton trees $T_{1,1}, \ldots, T_{k,1}$.
In each round we will combine pairs of trees, and each tree $T$ will have a designated node $\nu(T)$. 
In the initial trees the designated node is the only node.
In the first round we combine pairs of trees by performing 
$$\Unite{\nu(T_{1,1}), \nu(T_{2,1})}, \Unite{\nu(T_{3,1}), \nu(T_{4,1})}, \dots, \Unite{\nu(T_{k-1,1}), \nu(T_{k,1})}$$
to produce tree $T_{1,2}, \ldots, T_{k/2, 2}$.
The designated node $\nu(T_{i, 2})$ is chosen to be one of the designated nodes of the subtrees that formed $T_{i, 2}$.
We call this process of picking the new designated nodes as a subset of the old ones {\em refining}.
The subsequent rounds are done similarly by combining pairs of trees from the previous round and refining designated nodes,
until only the tree $T_{1, \lg k}$ remains.\par

\vspace{\baselineskip}
We now make the following observations about this process:
\begin{enumerate}
\item
All trees $T_{i, r}$ of a given round $r$ have the same number of nodes $2^r$.

\item
A designated node always has depth at most 2.
(This follows from the way \textit{find} does compactions.)

\item
A node of depth $\delta$ in any of the trees $T_{i, r}$ has at most $\left( \frac{1}{2} \right)^\delta \cdot |T_{i, r}|$ successors.
\end{enumerate}

\vspace{\baselineskip}
The links in the rounds raise the depth of half the nodes due to (1),
and the compactions in the \textit{find} operations of the rounds reduce the average depth of a node in the forest by at most $\frac{1}{4}$ due to (2) and (3).
Thus, each round increases the average depth of a node in the forest by at least $\frac{1}{2} - \frac{1}{4} = \frac{1}{4}$.
Since there are $\log(k)$ rounds, the proof is complete.
\end{proof}

\vspace{\baselineskip}
Combining the previous lemmas yields our best algorithmic lower bound result.\par

\vspace{\baselineskip}
\begin{lemma}
    Let $\mathcal{A}$ be a concurrent disjoint set union algorithm obtained by combining
	linking by DCAS, randomized linking by rank, or linking by random index with find with no compaction, one-try splitting, or two-try splitting.
	There is a schedule of $m$ operations on $n$ nodes by $p$ processes that forces $\mathcal{A}$ to perform
	$\Omega\left(m \log{\left(\frac{np}{m} + 1 \right)} \right)$ work.
	The bound holds in expectation for the linking by random index algorithm even under the independence assumption.
\label{lem:algorithmicLogLowerBound}
\end{lemma}

\begin{proof}
    We prove the theorem for linking by the DCAS and randomized linking by rank algorithms first.
    The lower bound is non-trivial only when $m/p < n$.
    In this case, we describe a particular sequence of operations and schedule that performs the requisite work.
	Divide the nodes into $m/p$ groups of size $n/(m/p) = np/m$.
	Lemma~\ref{lem:treeBuilder} allows us to link each group of nodes into a tree of height 
	$\Omega\left(m \log{\left(\frac{np}{m} + 1 \right)} \right)$.
	For each such tree, perform $\Find{x}$ on the deepest node $x$ of that tree simultaneously with each of the processes.
	Now consider the schedule in which processes shadow each other in all the finds.
	In this schedule, each process does $\Omega\left(m \log{\left(\frac{np}{m} + 1 \right)} \right)$ work per find, and one find per group.
    The total number of operations is $m/p \times p = m$, and the total amount of work is
	$\Omega\left(m \log{\left(\frac{np}{m} + 1 \right)} \right)$.\par

    \vspace{\baselineskip}	
	In the case of the linking by random index algorithm under the independence assumption, we modify the above argument by replacing the use of Lemma~\ref{lem:treeBuilder} with Lemma~\ref{lem:treeBuilder2},
	and performing $\Find{x}$ on a uniformly randomly picked node $x$ in the tree.
\end{proof}

\vspace{\baselineskip}
Combining the previous lemmas yields our best algorithmic lower bound result.\par

\vspace{\baselineskip}
\begin{theorem}
    Let $\mathcal{A}$ be a concurrent disjoint set union algorithm obtained by combining
	linking by DCAS, randomized linking by rank, or linking by random index with find with no compaction, one-try splitting, or two-try splitting.
	There is a schedule of $m$ operations on $n$ nodes by $p$ processes that forces $\mathcal{A}$ to perform
	$\Omega\left(m \left( \log{\left(\frac{np}{m} + 1 \right)} + \alpha{\left(n, \frac{m}{np} \right)} \right) \right)$ work.
	The bound holds in expectation for the linking by random index algorithm even under the independence assumption.
\end{theorem}

\begin{proof}
    Combine the results of Lemmas \ref{lem:concurrentAckermannBound} and \ref{lem:algorithmicLogLowerBound}.
\end{proof}

\vspace{\baselineskip}
As the final algorithmic lower bound, we prove that the independence assumption we have been using to analyze the linking by random index algorithm is indeed necessary.
In particular, we present a super-logarithmic work lower bound for the algorithm if the independence assumption does not hold.

\vspace{\baselineskip}
\begin{lemma}
	Concurrent set union via the linking by random index algorithm performs $\Omega(m\sqrt{p})$ expected work to do $m = n\sqrt{p}$ operations 
	if $\sqrt{p} \le n$, regardless of which compaction rule the \textit{find} operations use.
\label{lem:neccesity-of-independence-assumption}
\end{lemma}

\begin{proof}
We will show an explicit example.
Assume $\sqrt{p} < n$, and pick a set $S$ of $\sqrt{p}$ nodes.
Let $p/2$ processes attempt to do $\Unite{x,y}$ where each pair of $x,y$ in $S$ is tried by at least one process.
The scheduler can wait to see the outcomes of the node comparisons and decide to schedule the processes so that the nodes of $S$
get linked into a linear path of length $\sqrt{p}$.
If the remaining $p/2$ processes all perform $\Find{x}$ where $x$ is chosen randomly from $S$,
and are scheduled in lock-step, they will perform, in expectation, $\Omega(\sqrt{p})$ work each, since the expected depth of $x$ is $\sqrt{p}/2$.\par

\vspace{\baselineskip}
Performing the same process on each of the $\floor{\frac{n}{\sqrt{p}}}$ sets of nodes leads to $\Omega(np)$ work to do $m = n\sqrt{p}$ operations.
The average operation takes $\Omega(\sqrt{p})$ work.
\end{proof}

\subsection{Problem Lower Bounds}

\vspace{\baselineskip}
In this subsection, we prove that any concurrent set-union algorithm must do $\Omega(\log \min \{n, p\})$ work for a single operation in the worst case.
Furthermore, we build on the worst-case lower bound to show an $\Omega(\log (np/m + 1))$ amortized work lower bound for all {\em symmetric algorithms}, where
we say an algorithm is symmetric if:
\begin{enumerate}
    \item 
    The algorithm's code for the $\textit{unite}$ and $\textit{find}$ procedures does not use process ids.
    
    \item
    The algorithm does not use the return values of CAS operations.
\end{enumerate}

\vspace{\baselineskip}
All our algorithms and all algorithms known to us can be made symmetric without effecting the upper bound analyses of the algorithms.
For instance, this can be done if we assume that all CAS operations return false.
This does not effect the correctness of our algorithms since we only use the return value of a CAS operation in the $\textit{unite}$ operation to determine if a link has been successful.
However, if we do not perform this check atomically, our algorithms remain correct, since a process that has successfully united two trees together will realize this shortly afterwards when its $u$ and $v$ pointers meet at the root of the united tree.
The work efficiency analysis increases by at most a factor of two, because we can always imagine the case where processes work in pairs $(p,q)$, and each pair performs operations together and are always scheduled in lock-step.
In this case, if $p$ ever performs a successful link \CAS{$u.p,u,v$} and returns, then $q$'s attempt to perform the same link will fail;
thus $q$ will only return after it traverses the whole tree and discovers that some other process has already finished the link it wanted to do.
The modification we propose to symmetrize the algorithm will simply make $p$ do the same work as $q$.\par

\vspace{\baselineskip}
Our lower bounds make use of a result on a problem called ``wake-up".
The {\em wake-up} problem on $k$ processes asks for a wait-free algorithm with two properties: 
(i) every process returns a boolean value and at least one process returns true, and
(ii) a process may return true only if every process has already executed at least one step.
The following Lemma is a lower bound on the complexity of wake-up that follows straightfowardly from Jayanti's lower bound
in \cite{PJayanti}.\par

\vspace{\baselineskip}
\begin{lemma}[\cite{PJayanti98, PJayanti}]
	For any $k$ process wake-up algorithm that uses variables supporting read, write, and CAS, there is a schedule in which some process performs $\Omega(\log k)$ steps in expectation.
	\label{lem:wulowerbound}
\end{lemma}

\vspace{\baselineskip}
We solve wake-up via set-union to get our lower bounds below.

\vspace{\baselineskip}
\begin{lemma}
	Reduction \ref{alg:wakeup} $($below$)$ solves the wake-up problem for $k$ processes using a disjoint set union instance with $k+1$ nodes,
	in which each process executes one \textit{unite} and two \textit{find} operations.
\label{lem:wakeup}
\end{lemma}

\begin{proof}
	Let $q_1,\ldots,q_k$ be the $k$ processes and let the nodes be labelled $0,\ldots,k$.
	The reduction below correctly solves wake-up because:

\makeatletter
\renewcommand{\ALG@name}{Reduction}
\makeatother		
    \begin{algorithm}[H]
    \caption{: $q_j$'s code in wake-up solution.}
    \begin{algorithmic}[1]
    	\Procedure{WakeUp}{}
    		\State $\Unite{j-1, j}$; $x \gets \Find{0}$; $y \gets \Find{k}$; \Return{$x = y$}
    	\EndProcedure
    \end{algorithmic}
    \label{alg:wakeup}
    \end{algorithm}
    (i) the last process to complete \textit{unite} finds that the representatives of nodes 0 and $k$ are the same and thus returns true, and
	(ii) no process returns true before all processes have completed \textit{unite}, since no representative of $k$ can be the same as any representative of 0 until they are in the same set, i.e. until the last of the	$\textit{unite}$ operations is linearized.
\end{proof}

\vspace{\baselineskip}
\begin{theorem}
    Let $\mathcal{A}$ be a linearizable wait-free concurrent disjoint set union algorithm using read, write, and CAS.
	There is a schedule of $m$ operations on $n$ nodes by $p$ processes that forces $\mathcal{A}$ to perform
	$\Omega(\log \min\{n, p\})$ work in expectation.
    \label{thm:worstCaseWorkLowerBound}
\end{theorem}

\begin{proof}
    Instantiate Lemma \ref{lem:wakeup} with $k = \min\{n-1, p\}$.
    The most expensive of the three set union operations of the process that performs the most work in the adversarial schedule
    of Lemma \ref{lem:wulowerbound} must do $\Omega(\log \min\{n, p\})$ expected work.
\end{proof}

\vspace{\baselineskip}
\begin{corollary}
    The disjoint set union algorithm obtained by combining randomized linking with any form of find described in this paper gives an algorithm with optimal worst-case work per operation up to constant factors
    when $\log p = \Theta(\log n)$.
    \label{cor:worstCaseWorkLowerBound}
\end{corollary}

\vspace{\baselineskip}
\begin{remark}
    Theorem \ref{cor:worstCaseWorkLowerBound} shows that our set union algorithms with randomized linking have optimal work per operation when $p = n^\varepsilon$ for constant $\varepsilon$.
\end{remark}

\vspace{\baselineskip}
\begin{remark}
    Theorem \ref{thm:worstCaseWorkLowerBound} establishes a separation in worst-case work complexity between 
    sequential and concurrent set-union when $p = n^{\omega(\frac{1}{\log\log n})}$ since Blum's sequential set-union 
    algorithm has a worst-case work complexity of $O(\frac{\log n}{\log\log n})$ \emph{\cite{Blum}}.
\end{remark}
   
\vspace{\baselineskip} 
\begin{lemma}\label{lem:lowerbound}
    Let $\mathcal{A}$ be a linearizable wait-free symmetric concurrent disjoint set union algorithm using read, write, and CAS.
	There is a schedule of $m$ operations on $n$ nodes by $p$ processes that forces $\mathcal{A}$t to perform
    $\Omega(m \log(np/m + 1))$ work.
\end{lemma}

\begin{proof}
	Divide the $n$ nodes into $g = \frac{m}{p}$ groups of size $k+1$, where $k = \frac{np}{8m}$ (disregard any additional nodes);
	label the groups $G_1,\ldots,G_g$.
	Note that $m \ge p$ and $m \ge \frac{n}{2}$, so $k \le \frac{p}{4}$.
	We divide the $\frac{p}{2}$ (out of the $p$) processes into two sets 
	$A = \{q_1,\ldots,q_k \}$ and $B = \{q_{k+1},\ldots,q_{p/2} \}$.
	Note that $|B| \ge \frac{p}{4}$ and $|A \cup B| = \frac{p}{2}$.\par

    \vspace{\baselineskip}	
	Consider running the wake-up algorithm of Lemma \ref{lem:wakeup} on the $k$ processes in $A$ using the $k+1$ nodes in $G_1$.
	By Lemma \ref{lem:wulowerbound} there is a schedule $\sigma_1$ in which some process $q_i$ performs $\Omega(\log k)$ steps.	
	Assign to each process in $B$ the same sequence of three set union operations that $q_i$ performs, and
	define schedule $\sigma_1'$ to be the schedule $\sigma_1$, with the processes of $B$ interleaved in
	to run in lockstep with $q_i$.
	In this schedule, the processes $q_1,\ldots,q_{p/2}$ perform $\Omega(p \log k)$ work to do $p$ set union operations.
	Repeating this procedure on each group $G_j$ produces schedules $\sigma_j'$, each of which performs $\Omega(p \log k)$ work.
	Therefore, in the concatenated schedule of $\sigma_1'\sigma_2'\cdots\sigma_g'$, the processes $q_1,\ldots,q_{p/2}$ perform
	a total of $\Omega(gp \log (k + 1)) = \Omega(m \log(np/m + 1))$ work to do a total of $gp = m$ operations.
\end{proof}

\vspace{\baselineskip}
\begin{theorem}
    Let $\mathcal{A}$ be any linearizable wait-free symmetric concurrent disjoint set union algorithm using read, write, and CAS.
	There is a schedule of $m$ operations on $n$ nodes by $p$ processes that forces $\mathcal{A}$ to perform
	$\Omega\left(m \left( \log{\left(\frac{np}{m} + 1 \right)} + \alpha{\left(n, \frac{m}{np} \right)} \right) \right)$ work in expectation.
\label{thm:lowerbound}
\end{theorem}

\begin{proof}
    Combine the results of Lemma \ref{lem:concurrentAckermannBound}---whose argument applies to all symmetric algorithms---with Lemma \ref{lem:lowerbound}.
\end{proof}

\vspace{\baselineskip}
\begin{remark}
    Theorem~\ref{thm:lowerbound} shows that the set union algorithm obtained by combining randomized linking with two-try splitting has optimal amortized work efficiency amongst all symmetric algorithms $($up to a constant factor$)$.
\end{remark}

\vspace{\baselineskip}
The ideas behind our collection of upper and lower bounds lead us to make the following conjecture about the expected work complexity of concurrent disjoint set union.

\vspace{\baselineskip}
\begin{conjecture}
The expected work complexity of the concurrent set union problem is
$$\Theta\left(m \cdot \left(\log \left(\frac{np}{m} + 1\right) + \alpha\left(n, \frac{m}{np}\right)\right)\right).$$
\end{conjecture}

\vspace{\baselineskip}
In light of Theorem 6, which shows that randomized linking with two-try splitting satisfies the conjectured upper bound, a refutation of Conjecture 1 would imply a more efficient algorithm than randomized linking with two-try splitting.
On the other hand, a demonstration of the conjecture would involve proving a universal lower bound; namely, showing that Theorem 9 holds for {\em all} algorithms (as opposed to only symmetric ones).
While this paper proves a universal lower bound on the worst-case complexity of a single operation, it does not prove any universal lower bounds on the total work complexity of $m$ operations.
The only such lower bound that is known for the problem is exponentially weaker than the conjectured one.
We state this bound, by Jayanti et al., below.

\vspace{\baselineskip}
\begin{thm}[\cite{JayantiTarjanBoix} ]
Let $\mathcal{A}$ be any linearizable wait-free concurrent disjoint set union algorithm using read, write, and CAS.
There is a schedule of $m$ operations on $n$ nodes by $p$ processes that forces $\mathcal{A}$ to perform
$\Omega\left(m \left( \log\log{\left(\frac{np}{m} + 1 \right)} + \alpha{\left(n, \frac{m}{n} \right)} \right) \right)$ work in expectation.
\end{thm}

\label{sec:lowerbound}

\section{Remarks and Open Problems}

\vspace{\baselineskip}
We have presented three linking methods and two splitting methods for concurrent\ disjoint set union.  With any of the linking methods, with or without compaction, the number of steps per operation is $O(\log n)$, worst-case if linking is deterministic, high-probability if randomized.
With any of the linking methods and either of the splitting methods, the total work is $O(m(\alpha(n,d) + \log(1 + 1/d)))$, worst-case if linking is deterministic, average-case if randomized, where the problem density $d$ is $m/(np^2)$ if splitting is one-try, $m/(np)$ if splitting is two-try.
No matter what the density, the cost of concurrency is at most a factor of $\log p$, making our algorithms truly scalable.
The proofs of the bounds for linking by random index require assuming that the scheduler is non-adversarial, as discussed in Section 5.3.
The bounds differ for the two splitting methods only for a narrow range of densities: if $m/n = O(1)$ or $m/n = \Omega(p^2)$, the bounds are the same; if $m/n = \omega(1)$ and $m/n = o(p^2)$, the bounds differ by a factor of at most $\log p$.\  \par

\vspace{\baselineskip}
The $O(\log n)$ step bound is tight for all our algorithms. 
The work bounds for splitting are almost tight: any symmetric algorithm (as defined in Section 8) has a work bound of \linebreak
$\Omega\biggl( m \cdot \left( \log{\left(\frac{np}{m} + 1 \right)} + \alpha{\left(n, \frac{m}{np} \right)} \right) \biggr)$. 
We conjecture that the same lower bound can be shown for asymmetric algorithms also (Conjecture 1), but leave the proof or refutation of this statement as an open problem.
We also leave as an open problem finding an efficient {\em deterministic} algorithm that uses only CAS, not DCAS. 
We think the deterministic coin tossing method of Cole et al. \cite{CV86} can be used, but the resulting algorithm will be significantly more complicated than the one we have proposed, and the number of steps per operation and the total work bound will contain an additional factor of $\log^* p$.\par

\vspace{\baselineskip}
In some applications of disjoint set union, such as computing flow graph information \cite{tarjan1974finding, tarjan1973testing} each set has a name or some other associated value, such as the number of elements in the set.  We can extend the compressed tree date structure to support set values by storing these in the set roots.   In the sequential setting, it is easy to update set value information in $O(1)$ time during a link.  
But in the concurrent setting, updating the value in the new root during a link requires a DCAS or some more-complicated implementation using CAS.  
Updating root values using\ DCAS invalidates our analysis.  Consider $n$ singleton sets $\{1\}, \{2\},\ldots, \{n\}$.  
Suppose $p = n$, and unite(1, \textit{n}), unite(2, \textit{n}),$\ldots$, unite(\textit{n} – 1, \textit{n})\ are\ performed concurrently using linking by rank via DCAS.  Assume the tie-breaking total order is numeric.  At most one link will succeed initially, say the link of 1 and \textit{n}.\  After this link, all nodes except \textit{n} will still have rank 0, and \textit{n}\ will have rank 1.  The algorithm of Section 5.1 does all the remaining links concurrently using CAS, since none affects node \textit{n}.\  But if each such link needs to update the value in node \textit{n}, the remaining links must be done one-at-a-time, resulting in overall work $\Omega(np)$.\par

\vspace{\baselineskip}
We think this problem can be overcome, and that the concurrent set union problem with set values can be solved in a work bound that is quasi-linear in \textit{m} and poly-logarithmic in \textit{p}.\  But doing so may well require relaxing the linearization requirement: instead of continuing to try to link each node \textit{i} and \textit{n},\ suppose\ the\ algorithm does a different set of links that reduce the contention. 
For example, the algorithm could link roots in pairs, then the remaining roots in pairs, and so on.  The set resulting from all the links would be the same, but the intermediate sets would not correspond to any linearization of the original unites.  Even though it violates linearization, such an algorithm might suffice in many if not all applications.\par

\vspace{\baselineskip}
An\ algorithm of this kind needs a mechanism to restructure the links.  We think some sort of binary tree structure, like the one used by Ellen and Woelfel \cite{EW} in their fetch-and-increment algorithm but more dynamic, may suffice.  
We leave open the development of this idea or some other idea to solve the problem of concurrent sets with values.\par

\vspace{\baselineskip}
Finally, we think our methods will fruitfully extend to a distributed-memory, message-passing setting.

\end{document}